\documentclass[12pt]{article}
\usepackage{pifont}
\usepackage{epsfig}
\usepackage{dsfont}
\usepackage{amsfonts}
\usepackage{mathrsfs}
\usepackage{amssymb}
\usepackage{amsmath}
\usepackage{amsthm}
\usepackage{amscd}
\usepackage{titlesec}
\usepackage{url}
\usepackage{latexsym,color}
\usepackage{bm}
\usepackage{graphics}
\usepackage{subfigure,enumerate}
\usepackage{natbib}

\newcommand{\pkg}[1]{$\mathsf{ #1}$}

\newtheorem{lemma}{{Lemma}}
\newtheorem{theorem}{{Theorem}}
\newtheorem{definition}{{Definition}}
\newtheorem{remark}{{Remark}}

\newtheorem{corollary}{{Corollary}}

\makeatletter

\newcommand{\bZ}{\boldsymbol{Z}}
\newcommand{\mZ}{\mathcal{Z}}
\newcommand{\bS}{\boldsymbol S}
\newcommand{\mE}{\mathcal{E}}
\newcommand{\mX}{\mathcal{X}}
\newcommand{\mW}{\mathcal{W}}

\newcommand{\mR}{\mathcal{R}}

\newcommand{\bY}{\boldsymbol{Y}}
\newcommand{\SZ}{S_{\mathcal{Z}_n}}
\newcommand{\bu}{\boldsymbol{u}}
\renewcommand{\bm}{\boldsymbol{m}}
\newcommand{\btheta}{\boldsymbol{\theta}}
\newcommand{\bTheta}{\boldsymbol{\Theta}}
\newcommand{\bpi}{\boldsymbol{\pi}}

\newcommand{\ba}{\boldsymbol{a}}
\newcommand{\be}{\boldsymbol{e}}
\newcommand{\bz}{\boldsymbol{z}}
\newcommand{\bX}{\boldsymbol{X}}

\newcommand{\bSigma}{\boldsymbol{\Sigma}}
\newcommand{\bepsilon}{\boldsymbol{\epsilon}}
\newcommand{\bx}{\boldsymbol{x}}
\newcommand{\pl}{\tilde l}
\newcommand{\bb}{\boldsymbol{b}}
\newcommand{\bv}{\boldsymbol{v}}

\newcommand{\E}{\mathrm{E}}
\renewcommand{\P}{\mathrm{P}}

\newcommand{\bdelta}{\boldsymbol{\delta}}

\newcommand{\bLambda}{\boldsymbol{\Lambda}}
\newcommand{\by}{\boldsymbol{y}}

\newcommand{\bbeta}{\boldsymbol{\beta}}
\newcommand{\arglmin}{\mathop{\mathrm{arglmin}}}
\newcommand{\Rmnum}[1]{\expandafter\@slowromancap\romannumeral #1@}
\makeatother

\def\boxit#1{\vbox{\hrule\hbox{\vrule\kern6pt\vbox{\kern6pt#1\kern6pt}\kern6pt\vrule}\hrule}}

\begin{document}
\title{Likelihood Adaptively Modified Penalties}

\author{Yang Feng$^{1}$, Tengfei Li$^{2}$ and Zhiliang Ying$^{1}$\\
\small{$^{1}$Department of Statistics, Columbia University, New York, NY, USA. }\\
\small{$^{2}$Chinese Academy of Sciences, Beijing, China.}}

\date{\today}
\maketitle

%
%
%

\begin{abstract}
A new family of penalty functions, adaptive to likelihood, is
introduced for model selection in general regression models. It
arises naturally through assuming certain types of prior
distribution on the regression parameters. To study stability
properties of the penalized maximum likelihood estimator, two types
of asymptotic stability are defined. Theoretical properties,
including the parameter estimation consistency, model selection
consistency, and asymptotic stability, are established under
suitable regularity conditions. An efficient coordinate-descent
algorithm is proposed. Simulation results and real data analysis
show that the proposed method has competitive performance in
comparison with existing ones.

\end{abstract}

%

\section{Introduction}

Classical work on variable selection dates back to \citet{aic}, who
proposed to choose a model that minimizes the Kullback- Leibler (KL)
divergence of the fitted model from the true model, leading to the
well-known Akaike Information Criterion (AIC). \cite{bic} took a
Bayesian approach by assuming prior distributions with nonzero
probabilities on
 lower dimensional subspaces. He proposed what is known as the BIC method for model selection. Other types of $L_0$
penalties include $C_p$ \citep{CP}, AICC \citep{Hurvich.Tsai.1989}, RIC \citep{ric} and EBIC
\citep{ebic}, among others.

The $L_0$ regularization has a natural interpretation in the form of
best subset selection. It also exhibits good sampling properties
\citep{Barronbirge1999}. However, in a high-dimensional setting, the
combinatorial problem has NP-complexity, which is computationally
prohibitive. As a result, numerous attempts have been made to modify
the $L_0$ type regularization to alleviate the computational burden.
They include bridge regression \citep{bridge},  non-negative garrote
\citep{garrote95}, LASSO \citep{Tibshirani96}, SCAD \citep{Fan01},
elastic net \citep{elasticnet},  adaptive LASSO or ALASSO
\citep{Huizou06}, Dantzig selector \citep{Candes.Tao.2007}, SICA
\citep{Lv.Fan.2009}, MCP \citep{MCP}, among others.

To a certain extent, existing penalties can be classified into one
of the following two categories: convex penalty and nonconvex
penalty. Convex penalties, such as LASSO \citep{Tibshirani96}, can
lead to a sparse solution and are stable as the induced optimization
problems are convex. Nonconvex penalties, such as SCAD \citep{Fan01}
and MCP \citep{MCP}, can on the other hand lead to sparser solutions
as well as  the so-called oracle properties (the estimator works as
if the identities of nonzero regression coefficients were known
beforehand). However, the non-convexity of the penalty could make
the entire optimization problem nonconvex, which in turn could lead
to a local minimizer and the solution may not be as stable as the
one if instead a convex penalty is used. Therefore, an important
issue for nonconvex penalties is a good balance between sparsity and
stability. For example, both SCAD and MCP  have an extra tuning
parameter which regulates the concavity of the penalty so that, when
it exceeds a threshold, the optimization problem becomes convex.

It is well known that penalty functions have Bayesian
interpretation. The classical $L_2$ penalty (ridge regression) is
equivalent to the Bayesian estimator with a normal prior. The
$L_1$-type penalties, such as LASSO, ALASSO etc., also have Bayesian
counterparts; cf. \cite{Park.Casella.2008},
\cite{Griffin.Brown.2010} and \cite{Hara.Sillanp.2009}.

\cite{heuristic stability} initiated the discussion about the
issue of stability  in model selection. He demonstrated that many
model selection methods are unstable but can be stabilized by
perturbing the data and averaging over many predictors.
\cite{Breiman.2001} introduced the random forest, providing a way
to stabilize the selection process. \cite{Buhlmann.Yu.2002}
derived theoretical results to analyze the variance reduction
effect of bagging in hard decision problems.
\cite{Meinshausen.Buhlmann.2010} proposed a stable selection,
which combines the subsampling with high-dimensional variable
selection methods.

The main objective of the paper is to introduce a family of
penalty functions for generalized linear models that can achieve a
proper balance between sparsity and stability. Because for the
generalized linear models, the loss function is often chosen to be
the negative log-likelihood, it is conceivable to take into
consideration the form of the likelihood in the construction of
penalty functions. The Bayesian connection to the penalty
construction and to the likelihood function make it natural to
introduce penalty functions through suitable prior distributions.
To this end, we introduce the family of negative absolute priors
(NAP) and use it to develop what to be called likelihood
adaptively modified penalties (LAMP). We define two types of
asymptotic stability that cover  a wide range of situations and
provide a mathematical framework under which the issue of
stability can be studied rigorously. We show that under suitable
regularity conditions, the LAMP results in an asymptotically
stable estimator.


%


%

The rest of the paper is organized as follows. Section 2
introduces the LAMP family with motivations from its likelihood
and Bayesian connections. Specific examples are given for the
commonly encountered generalized linear models. In Section 3, we
introduce two types of asymptotic stability and study the
asymptotic properties of LAMP family. The choice of the tuning
parameters and an efficient algorithm are discussed in Section 4.
In Section 5, we present simulation results and applied the
proposed method to two real data sets. We conclude with a short
discussion in Section 6. All the technical proofs are relegated to
the Appendix.

\section{Likelihood adaptively modified penalty}\label{sec:Lamp:Lamp}

To introduce our approach, we will first focus on the generalized
linear models \citep{Nelder.Wedderburn.1972}. It will be clear
that the approach also works for other types of regression models,
including various nonlinear regression models. Indeed, our
simulation results presented in Section 5 also include the probit
model, which does not fall into the exponential family induced
generalized linear models.

Throughout the paper, we shall use $(\bX,Y)$ and $(\bX_i,Y_i)$ to
denote the iid random variables for $i=1,2,\cdots,n$, where $Y_i$ is
the response observation of $Y$ and $\bX_i$ is the $p+1$-dimensional
covariate observation of $\bX$, $\bX_i=(X_{i0},X_{i1},$
$\cdots,X_{ip})^T$ with $X_{i0}\equiv1$. Let
$\mX=(\bX_1,\bX_2,\cdots,\bX_n)^T$, the $n\times(p+1)$ matrix of
observed covariates.  Following \cite{Nelder.Wedderburn.1972}, we
assume that the conditional density of $Y_i$  given covariates
$\bX_i$ has the following form
\begin{align}\label{eq:glm:def}
f(Y_i,\btheta|\bX_i) = h(Y_i)
\exp\left[\frac{Y_i\xi_i-g(\xi_i)}{\varphi}\right],
\end{align}
where $g$ is a smooth convex function,
$\boldsymbol\theta=(\alpha,\boldsymbol\beta^T)^T=(\theta_0,\theta_1,\cdots,\theta_p)^T$
or $(\theta_0,\beta_1,\cdots,\beta_p)^T$,
$\xi_i=\bX_i^T\boldsymbol\theta$ and $\varphi$ is the dispersion
parameter. Then up to an affine transformation, the log-likelihood
function is given by
\begin{align}
l(\btheta) = \sum_{i=1}^n [Y_i\xi_i-g(\xi_i)].
\end{align}
Note that the form of $l$ is uniquely determined by $g$ and vice
versa.

For a given $g$, we propose the induced penalty function
\begin{equation}\label{lamp}
p_{\lambda}(\beta)=\frac{\lambda^2}{g'(\alpha_1)\lambda_0}
[g(\alpha_1)-g(\alpha_1-\frac{\lambda_0}{\lambda}\beta)],
\end{equation}
which contains three parameters $\alpha_1\leq 0$, $\lambda_0>0$ and
$\lambda>0$. The corresponding negative penalized log-likelihood
function is
\begin{align}\label{eq::lik-pen}
-\pl(\btheta)=-l(\btheta)+n\sum_{j=1}^pp_{\lambda}(|\beta_j|),
\end{align}
Because the penalty function defined by (\ref{lamp}) is likelihood
specific, we will call such a penalty likelihood adaptively modified
penalty (LAMP).

We clearly have $p_{\lambda}(0)=0$ and
$p'_{\lambda}(0)=\lambda$. Furthermore, taking the first and
second derivatives, we get
$$p'_{\lambda}(\beta)=\lambda\frac{g'(\alpha_1-\frac{\lambda_0}{\lambda}
\beta
)}{g'(\alpha_1)} \mbox{ and }
p''_{\lambda}(\beta)=-\lambda_0\frac{g''(\alpha_1-\frac{\lambda_0}{\lambda}\beta
)}{g'(\alpha_1)}.$$

\begin{remark}
The parameters have clear interpretations: $\lambda$ is the usual
tuning parameter to control the overall penalty level; $\alpha_1$
is a location parameter, which may be fixed as a constant;
$\lambda_0$ controls the concavity of the penalty in the same way
as $a$ in SCAD (Fan and Li, 2001) and $\gamma$  in MCP (Zhang,
2010).
\end{remark}

\begin{remark}
The exponential family assumption given by \eqref{eq:glm:def} implies that
$g'(\xi)=E(Y|\bX)$ and $g''(\xi)>0$. When $Y$ is nonnegative,
$g'(\xi)\ge 0$. Thus $p_\lambda''$ is negative and the
corresponding LAMP must be concave.
\end{remark}

%

Like many other penalty functions, the family of LAMPs also has a
Bayesian interpretation. To see this, we introduce the following
prior for any given $g$ that defines the exponential family
\eqref{eq:glm:def}.

\begin{definition}
Let $a_j\leq 0$, $b_j<0$ and $c_j,$ $j=1,\cdots,p$ be constants.
Define a prior density, to be called negative absolute prior
(NAP),
  \begin{align}
   p(\boldsymbol\beta)\propto \prod_{j=1}^p\exp[-c_j\{g(a_j)-g(a_j+b_j|\beta_j|)\}].
  \end{align}
\end{definition}

If we choose $a_j=\alpha_1$, $b_j=-\lambda_0/\lambda$ and
$c_j=-\lambda^2/[g'(\alpha_1)\lambda_0]$, then the posterior mode
is exactly the minimizer of (\ref{eq::lik-pen}) the penalty form
of LAMP. We will show that such a choice will lead to good
asymptotic properties.

\begin{remark}
The additive form for the penalty function suggests that the prior
must be of the product form. For each $i$, hyperparameter $b_i$
scales $\beta_i$ while hyperparameter $c_i$ gives the speed of
decay. For the $i$th parameter $\beta_i$, the larger the values of
$b_i$ and $c_i$ are, the more information the prior has for
$\beta_i$. Translating it into the penalty function, it means that
values of $b_i$ and $c_i$ represent the level of penalty and they
can be adjusted separately for each component.
\end{remark}

\begin{remark}
The form of NAP is to similar to that of a conjugate prior, in
that the shape of the prior is adapted to that of the likelihood
function. However, NAP also differs from the conjugate prior in
several aspects: negative when $g'(\cdot)>0$, absolute, from
separate dimension, and with a LASSO term taken away. Unlike
conjugate prior, which assumes additional samples, this can be
seen as to take away the absolute value of redundant sample
information from each dimension.
\end{remark}

%
%
It is worth looking into the commonly encountered generalized linear
models and examine properties of the corresponding LAMPs.\\

{\it Linear regression}. In this case, $g(\xi)=\xi^2/2$. Thus, LAMP
reduces to the elastic net \citep{elasticnet}.
$$p_{\lambda}(\beta)= \lambda \beta-\alpha_1^{-1}\frac{\lambda_0}{2}\beta^2.$$

{\it Logistic regression}. Here  $g(\xi)=\log(1+e^{\xi})$.
Consequently, the penalty function
  \begin{align}\label{eq::sigmoid}
p_{\lambda}(\beta)=\frac{\lambda^2(1+\rho)}{\lambda_0\rho}
\log\left[\frac{1+\rho}{1+\rho
\exp(-\lambda_0/\lambda\beta)}\right],
  \end{align}
where $\rho=\exp(\alpha_1)>0$. This penalty will be called sigmoid penalty.

{\it Poisson regression}. Here $g(\xi) =e^{\xi}$ and we have
$$
p_{\lambda}(\beta)=\frac{\lambda^2}{\lambda_0}[1-\exp(-\frac{\lambda_0}{\lambda}\beta)].
$$
This will be called the Poisson penalty.

{\it Gamma regression}. For the gamma regression,
we have $g(\xi)=-\log(-\xi)$. Then, the penalty has the following
form
$$p_{\lambda}(\beta)=\frac{\lambda^2\alpha_1}{\lambda_0}[\log(-\alpha_1)-\log(\frac{\lambda_0}{\lambda}\beta-\alpha_1)].$$

{\it Inverse gaussian regression}. For the inverse Gaussian, we have $g(\xi)=-\sqrt{-2\xi}$. The resulting
penalty
$$p_{\lambda}(\beta)=\frac{2\lambda^2}{\lambda_0}(-\alpha_1)^{1/2}[(\frac{\lambda_0}{\lambda}\beta-\alpha_1)^{\frac{1}{2}}-(-\alpha_1)^{1/2}].$$

{\it Probit regression}. As we mentioned earlier, LAMP approach can
also accommodate regression models not in the form of naturally parametrized
exponential families. In
particular, it is applicable to the probit regression for binary
outcomes. In this case, $g(\xi)=-\log(\Phi(-\xi))$, which leads to
the following penalty form
$$p_{\lambda}(\beta)=\frac{\lambda^2}{\lambda_0}\frac{\phi(-\alpha_1)}{\Phi(-\alpha_1)}{\log[\Phi(-\alpha_1+\frac{\lambda_0\beta}{\lambda})-\Phi(-\alpha_1)]},$$
where $\Phi(\cdot)$ and $\phi(\cdot)$ are respectively distribution and density functions of the standard normal random variable.

\begin{remark}\label{remark::alpha_1_choice}
For the above examples, the effect of tuning parameters in the
penalty function varies across settings. For example, $\alpha_1$
does not play a role in the Poisson regression. In addition, we
have a natural choice for $\alpha_1$ for the penalty functions.
Specifically, we can choose $\alpha_1=-1$ for the cases of linear,
gamma, and inverse gaussian, and $\alpha_1=0$ for the cases of
logistic and probit.
\end{remark}

\begin{figure}
\caption{Penalties and the derivative of the penalties.\label{fig:penaltyfunctions1}}
\includegraphics[scale=0.35]{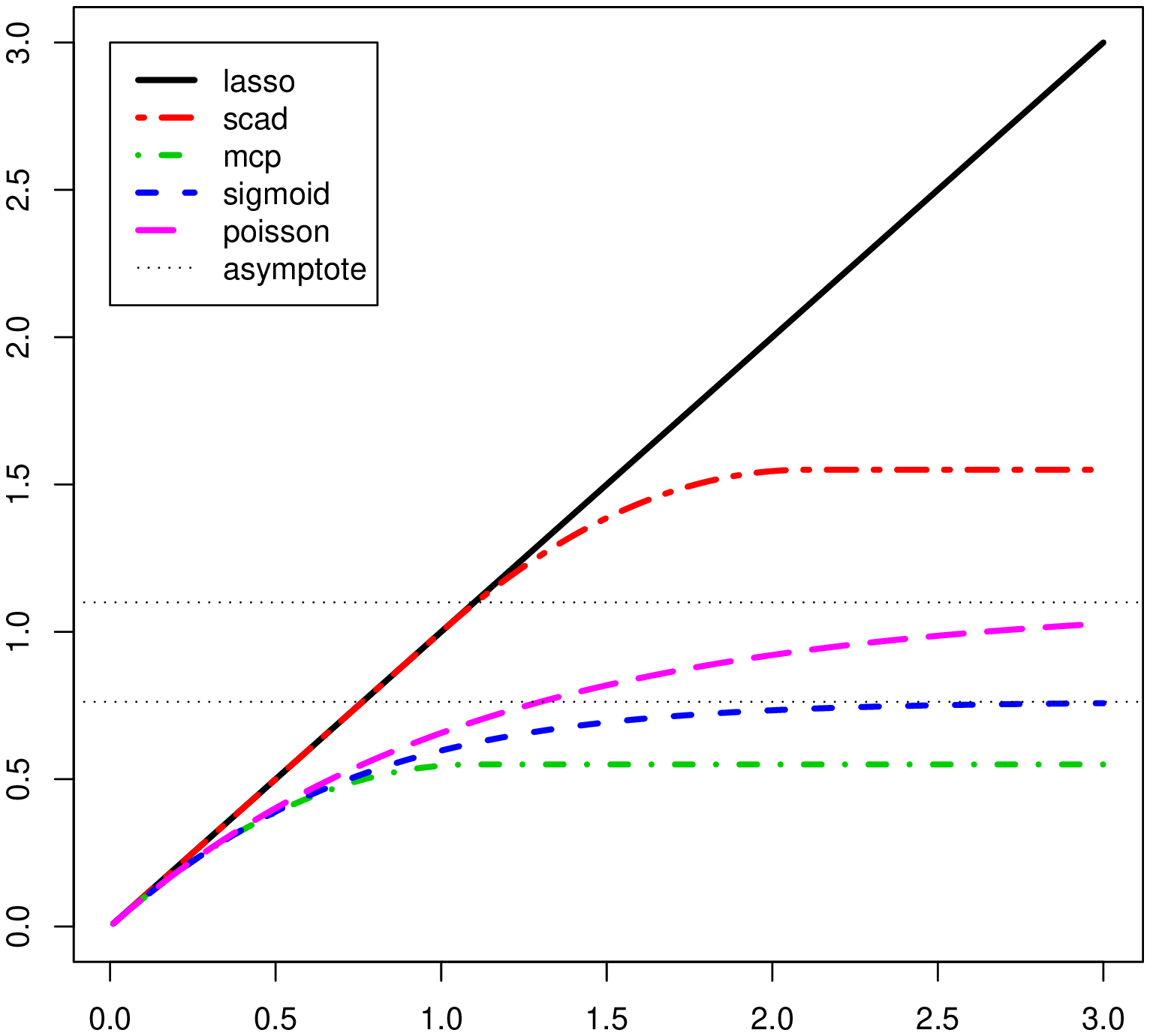}
\includegraphics[scale=0.35]{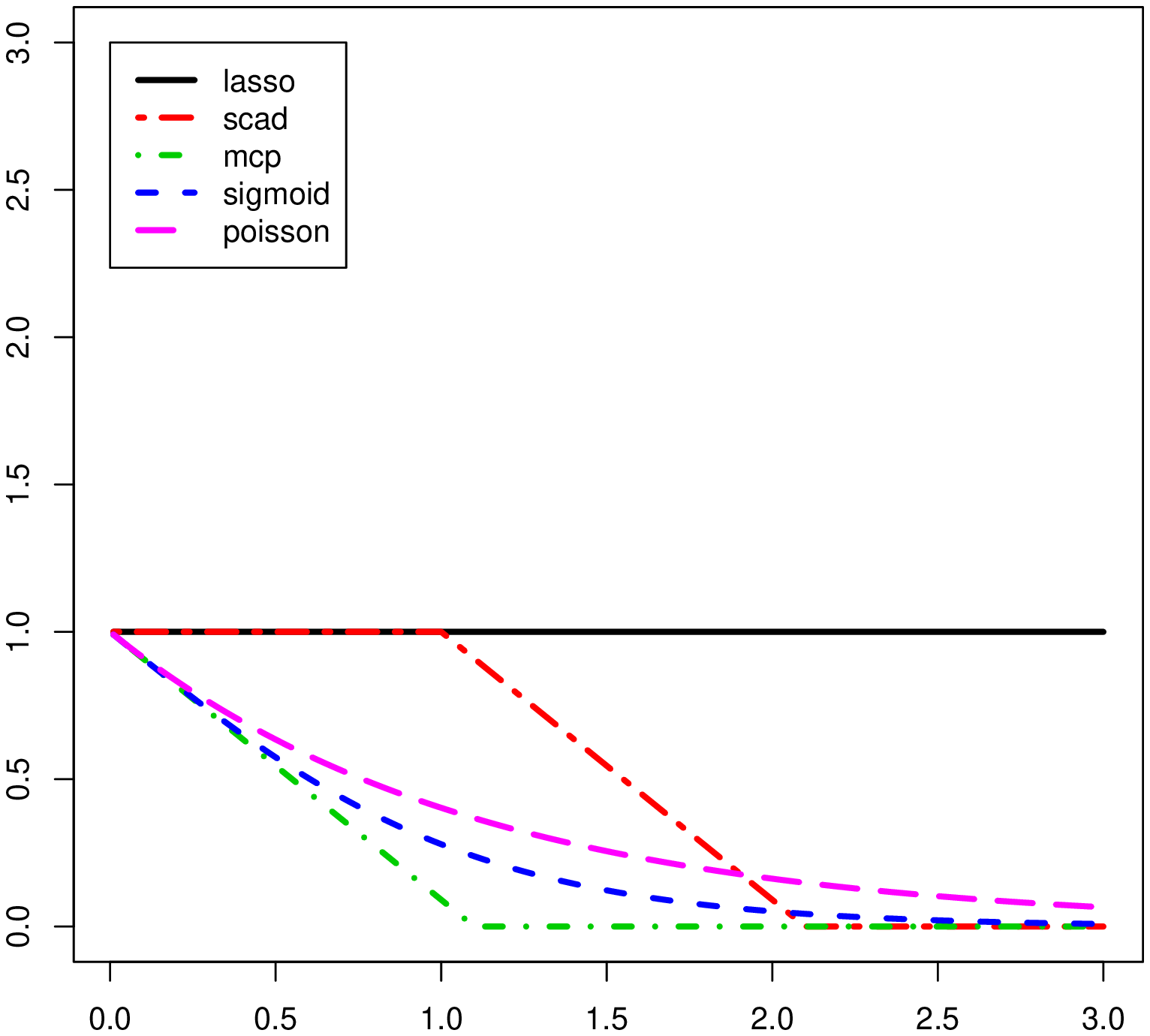}
\end{figure}

As the examples show, the LAMP family is fairly rich. They also
differ from the commonly used penalties. Figure
\ref{fig:penaltyfunctions1} contains plots of penalty functions
(left panel) along with their derivatives (right panel) that include
LASSO, SCAD, MCP and two members of the LAMP family (sigmoid and
Poisson). Here $\gamma=1.1$ for MCP, $a=2.1$ for SCAD,
$\lambda_0=2/1.1\approx1.82$ for sigmoid and $1/1.1\approx0.91$ for
Poisson penalty, and $\rho=1$ for sigmoid penalty. The parameters
are chosen to keep the maximum concavity of these penalties the
same.
 Figure \ref{fig:penaltyfunctions1} shows that sigmoid and Poisson penalties lie between MCP and SCAD when the maximum concavity is the same. Also, from the graphs of the derivatives, it is easy to identify that the penalties of the LAMP family have continuous derivatives (actually they have continuous derivatives of any order for most common generalized linear models) as compared with the discontinuous ones for SCAD and MCP. It will be shown that this feature can make the optimization problem easier and the estimation more stable.
 \begin{remark}
Consider a one-dimensional penalized logistic regression problem,
$\min_{\theta}[-l(\theta)+p_{\lambda}(|\theta|)]$, where
derivative of the penalty function $p'_{\lambda}(\cdot)>0.$ If the true
parameter $\theta_0\neq0$, we need
$p'_{\lambda}(|\theta_0|)\le\sup_{\theta\in U(\theta_0,\epsilon)}l'(\theta)$, where $U(\theta_0;\epsilon)$ is a $\theta_0$-centered ball with radius  $\epsilon>0$ uniform in $n$,   to get  $\hat\theta\neq0$, while  for $\theta_0=0$, we need
$p'_{\lambda}(0)\ge l'(0)$. Thus, to differentiate a nonzero $\theta_0$
from 0, we should make the difference between
$p'_{\lambda}(|\theta_0|)$ for $\theta_0\neq 0$ and $p'_{\lambda}(0)$ as large as
possible. In addition, we need to control the second derivative of the
penalty function to make the penalized negative log-likelihood function
globally convex. Figure \ref{fig:concavity} illustrates graphically how such dual objectives can be met for the logistic regression under the sigmoid penalty and MCP. The grey area represents the difference between $p'_{\lambda}(|A|)$ and
$p'_{\lambda}(0)$ for the sigmoid penalty. A subset of the grey area, marked by darker color, represent the corresponding difference under MCP. The reason that MCP tends to have a smaller difference compared with sigmoid penalty is because the second derivative of the penalty needs to be controlled by that of the negative log-likelihood in order to maintain the global convexity.

\end{remark}
\begin{figure}[!t]
\caption{Negative second derivative of log-likelihood, the sigmoid penalty and MCP in logistic
regression.\label{fig:concavity}}
\begin{center}

\includegraphics[scale=0.6]{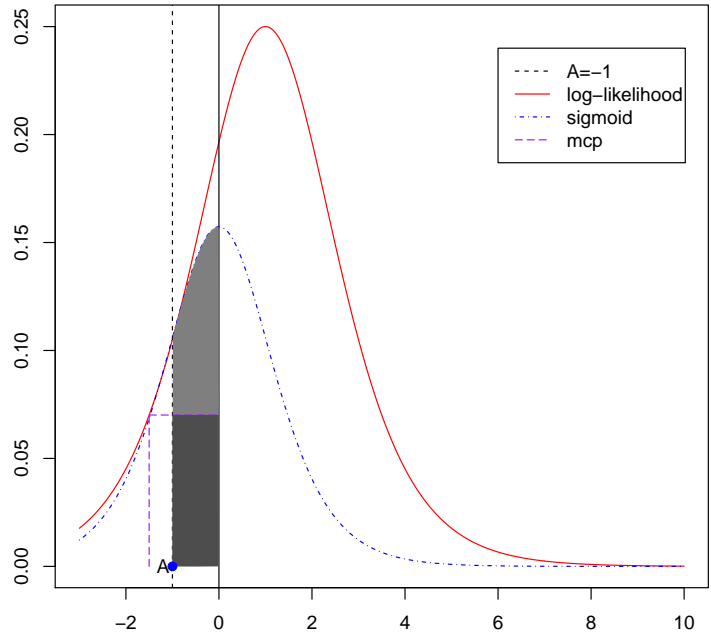}

\end{center}

\end{figure}

\section{Theoretical properties}
\subsection{Asymptotic stability}



Recall that the negative log-likelihood considered here is convex,
and the maximum likelihood estimation is uniquely defined and
stable. By adding a nonconvex penalty, the convexity may be
destroyed. To study stability of the penalized MLE, it is necessary
to study the impact of the penalty on the concavity, especially when
$n$ is large.

For the negative penalized maximum log-likelihood estimation
procedure \eqref{eq::lik-pen}, if nonconvex penalties are used, the
solution to the minimization problem may not be unique. Therefore,
it is natural to study the behavior of the local maximizers in
penalized likelihood estimates when the observations are perturbed.
Here we introduce a new concept of asymptotic stability to describe
the asymptotic performances of local minimizers in penalized
likelihood estimates. Note that even for the negative penalized
 likelihood estimators with convex penalty where the unique
minimizer exists, such asymptotic stability concept is still useful
in characterizing the behavior of the global maximizer.


Suppose we want to minimize with respect to $\btheta$ a criterion
function $M_n(\mZ_n,\btheta)$, where $\mZ_n=(\bZ_1,\cdots,\bZ_n)^T$
and $\bZ_i=(\bX_i^T, Y_i)^T$ is the $i$th observation of
$\bZ=(\bX,Y)$. Denote by $\bS_Z,S_{\mathcal{Z}_n}$ and $\bTheta$ the
support for $\boldsymbol Z, \mZ_n$ and domain for $\btheta$,
respectively. We say that $\btheta^*$ is a local minimizer if there
exists a neighborhood within which $M_n(\mZ_n,\cdot)$ attains its
minimum. More precisely, the set of the
local minimizers is defined as
\begin{align*} \arglmin\limits_{\bTheta}
M_n(\mZ_n,\btheta)\triangleq&
\{\btheta^*\in\bTheta|\exists\epsilon>0, M_n(\mZ_n,\btheta^*)=\min_{
\|\btheta-\btheta^*\|\leq\epsilon} M_n(\mZ_n,\btheta)\}.
\end{align*}
Throughout the paper, $\|\cdot\|$ denotes the $L_2$-norm of a vector
or matrix (vectorized) while $|\cdot|$ denotes the $L_1$ norm. For
any $ \btheta_*\in\bTheta,r>0$ let
$U(\btheta_*;r)\triangleq\{\btheta\in\bTheta|\|\btheta-A\|<r\}$ the be
$\btheta_*$-centered ball with radius $r$.

 It is clear from the definition that the set of local maximizers
is random. We characterize its asymptotic behavior in terms of
whether or not the set converges to a single point as $n\rightarrow
\infty$. For a set $A$, define its diameter as
$\mathrm{diam}(A)\triangleq\sup\limits_{\bx,\by\in A}\|\bx-\by\|$.

\begin{definition}[Weak asymptotic stability]
We say that the set of local minimizers of $M_n(\bZ_n,\cdot)$
satisfies weak asymptotic stability if $\forall \delta>0,$
\begin{equation}\label{stab}
    \lim_{n\rightarrow\infty}\mathrm{P}\Big(\varlimsup_{\epsilon\rightarrow0}
\mathrm{diam}\Big[\bigcup_{\substack{\|\mE_n\|/\sqrt{n}<\epsilon\\
\mathcal{Z}_n+\mathcal{E}_n\in
\SZ}}\{\arglmin_{\btheta}M_n(\mZ_n+\mE_n,\btheta)\}\Big]>\delta\Big)=0.
\end{equation}

\end{definition}

The weak asymptotic stability characterizes the asymptotic behavior
of local minimizers when the data are perturbed slightly. It shows
that for large $n$ and small perturbation, the local minimizers stay
sufficiently close to each other with high probability. Defined
below is a stronger version, which guarantees uniqueness of the
minimizer.

\begin{definition}[Strong asymptotic stability]
We say that the set of local minimizers of $M_n(\bZ_n,\cdot)$
satisfies strong asymptotic stability if
\begin{equation}\label{stab2}
    \lim_{n\rightarrow\infty}\mathrm{P}\Big(\varlimsup_{\epsilon\rightarrow0}
\mathrm{diam}\Big[\bigcup_{\substack{\|\mE_n\|/\sqrt{n}<\epsilon\\
\mathcal{Z}_n+\mathcal{E}_n\in
\SZ}}\{\arglmin_{\btheta}M_n(\mZ_n+\mE_n,\btheta)\}\Big]=0\Big)=1.
\end{equation}
\end{definition}

\begin{remark}
Under the weak asymptotic stability, multiple minimizers, though
shrinking to $0$, may exist with high probability for any finite
$n$. The strong asymptotic stability, on the other hand, entails
that for sufficiently large $n$, the probability of having multiple
minimizers must converge to $0$, implying that there must be a
unique minimizer with high probability.
\end{remark}
\begin{remark}
$L_0$ penalties, though may have the weak asymptotic property if we
adjust its tuning parameter like AIC, will never possess the strong
asymptotic property, because then the solution of each submodel with
the number of parameters constrained will be a local optimizer, and
with no probability all solutions will coincide, as
$n\rightarrow\infty$.
\end{remark}
We now consider the situation in which $M_n(\mZ_n,\btheta)$ can be
approximated by an i.i.d. sum with a remainder term, i.e.,
\begin{align}\label{eq:stability}M_n(\mZ_n,\btheta)=\frac{1}{n}\sum_{i=1}^n
m(\bZ_i,\btheta)+r_n(\btheta).
\end{align}
This general form includes the form of a negative log-likelihood
plus a penalty. Let $m_n(\mZ_n,\btheta)\triangleq
n^{-1}\sum_{i=1}^nm(\bZ_i,\btheta)$. We assume throughout the rest
of the paper that $\bTheta$ is compact with the true parameter
value, denoted by $\btheta_0$, lying in its interior and that
$\bar m(\btheta)=\E m(\bZ,\btheta)$ is finite. We need the
following regularity conditions.
\begin{enumerate}
  \item [(C1)] For any
$\bx\in \bS_Z, m(\bx,\btheta)$ is convex as a function $\btheta$,
and $\bar m(\btheta)$ is
 strictly convex.
\item [(C2)]%
There exists function $K: \bS_Z\mapsto \mathds{R},$ such that $\E
K^2(\bZ)<\infty$ and, for any $\bdelta_1,$ $\bdelta_2$ and $\bz$,
such that $\bz+\bdelta_i\in \bS_Z, i=1,2,$
$$\sup_{\btheta\in\bTheta}|m(\bz+\bdelta_1,\btheta)-m(\bz+\bdelta_2,\btheta)
|\leq K(\bz)\|\bdelta_1-\bdelta_2\|.$$
\item [(C3)] The remainder term is asymptotically flat:
$$\lim_{n\rightarrow \infty}\sup_{\btheta_1,\btheta_2\in\bTheta}\frac{|r_n(\btheta_1)-r_n(\btheta_2)|}{\|\btheta_1-\btheta_2\|}
=0.$$
\item [(C4)]
There exists a local minimizer of \eqref{eq:stability} that is a
consistent estimator of $\btheta_0.$
\item [(C5)]
There exists $\delta_0>0,$ such that
$$ \lim_{n\to\infty} P\left(M_n(\mZ_n,\btheta) \mbox{~is strictly convex within~} U(\btheta_0;\delta_0)\cap\bTheta\right)=1.$$
\end{enumerate}

\begin{remark}
When $m$ and $r_n$ are both smooth functions, Conditions (C2) and
(C3) can be guaranteed by assuming that the derivative of $m$ is
bounded uniformly and that the derivative of $r_n$ tends to 0
uniformly. Condition (C5) is guaranteed by the convexity around the
true parameters, uniformly in $n$.
\end{remark}

The next two lemmas provide sufficient conditions for  the two types of  asymptotic stability.

\begin{lemma}\label{lemma::stability:weak}
If  Conditions (C1)-(C3)  are satisfied, then we have weak asymptotic
  stability.
\end{lemma}

It is straightforward to verify that for generalized linear models, SCAD, MCP, sigmoid penalty and Poisson penalty all
satisfy Conditions (C1)-(C3), leading to the weak asymptotic stability.

\begin{lemma}\label{lemma:stability:strong}
If Conditions (C1)-(C5) are satisfied, then strong asymptotic
stability holds.
\end{lemma}

\subsection{Asymptotic properties}
In this subsection, we  study asymptotic properties for the proposed
LAMP, including parameter estimation consistency, model selection
consistency and asymptotic stability. Suppose $\beta_{j,0}$ or
$\theta_{j,0}$ is the true value of $\beta_j$ or $\theta_j$,
$\forall j\in \mathds{N}$. Let $q\triangleq
\sum_{j=1}^p\mathrm{\boldsymbol I}_{\beta_{j,0}\neq0}$, the number
of signals. Without loss of generalization, we assume
$\bbeta=(\bbeta_1^\tau,\bbeta_2^\tau)^\tau$, where the true value
$\bbeta_{1,0}$ has no zero element, and $\bbeta_{2,0}=\boldsymbol
0_{p-q}$, which indicates a $(p-q)\times1$ vector with each element
being $0$. Denote
$\zeta_1=\inf\{|\beta_{1,0}|,\cdots,|\beta_{q,0}|\}$,
$\zeta_2=\sup\{|\beta_{1,0}|,\cdots,|\beta_{q,0}|\}$, then
$[\zeta_1,\zeta_2]$ is the range of signal level. For any sequences
$a_n, b_n$ we say $a_n\ll b_n$ if
$\lim\limits_{n\rightarrow\infty}a_n/b_n=0$. Here, we consider the
setting of fixed $p,q,\zeta_1,\zeta_2$ when $n\rightarrow \infty$
though some results may be extended to the case of $p=o(\sqrt n)$.

We first introduce certain regularity conditions which are needed
for establishing asymptotic properties.
\begin{enumerate}
  \item[(C6)]$$\sup_{\stackrel{\btheta\in\bTheta}{0\le k\le
3}}E(\|\bX\|^k|g^{(k)}(\bX^{T}\btheta)|)<\infty, \inf_{{\btheta\in\bTheta}}\lambda_{\min}(\E
\bX g''(\bX^{T}\btheta)\bX^{T})>0,$$
 where $g^{(k)}$ denotes the $k$th derivative of $g$, and $\lambda_{\min}(\cdot)$ of a matrix ``$\cdot$'' denotes its
minimum eigenvalue.
\item[(C7)]  $g'(x)>0$ at $x\in[-\infty,\alpha_1)$ and
$\lim_{\xi\rightarrow-\infty}g'(\xi)=0$; $g''(\xi)$ is increasing at
$[-\infty,\alpha_1)$, where $\alpha_1\leq 0$ is a constant.
\item[(C8)] $\lambda\rightarrow0$,
$\sqrt{n}\lambda\rightarrow+\infty$, and $\sqrt{n}\lambda
g'(\alpha_1-\lambda_0 \zeta_1/\lambda)\rightarrow0.$
 \item[(C9)]
 $$\lambda_{\min}\{
\mathrm{E}[\bX\bX^T]\}>\frac{\lambda_0g''(\alpha_1)}{|g'(\alpha_1)|}\frac{1}{g''(\bX^T\boldsymbol\theta_0)}.$$
\end{enumerate}

\begin{theorem}\label{theo:LAMPasymp}
Suppose that Conditions (C6)-(C9) are satisfied. Then the penalized
maximum likelihood estimator based on the LAMP family is consistent
and asymptotically normal, and achieves model selection consistency
and strong asymptotic stability.

\end{theorem}
\begin{lemma}\label{lemma:rate}
Suppose that (C6) and (C7) are satisfied and
$\log[g'(-x)]=x^uL(x)$, where $u>0$ is a constant and $L(x)$ is
negative and slowly varying at $\infty$, i.e., for any $a>0$,
$\lim_{x\rightarrow \infty}[L(ax)/L(x)]=1,$ and
$\lim_{x\rightarrow+\infty}L(x)\in[-\infty,0)$.  Then $n^{-1/2}\ll
\lambda\ll \lambda_0(\log n)^{-1/u}\ll 1$ implies (C8).

\end{lemma}

It can be verified that the condition on $\log[g'(-x)]$ is
satisfied for the logistic, Poisson and probit regression models.
To achieve model selection consistency, we can choose
$\lambda_0=o(1)$ and  $\lambda=o((\log n)^{-1/u} \lambda_0).$

We now present the corresponding results for the case of sigmoid
penalty. Here, $\log(g'(-x))=xL(x),$ where $L(x)=[-\log(1+e^x)/x]$
is a slowly varying function. The following conditions simplify
(C6)-(C9).

\begin{enumerate}
\item[(C6')] $\E \|\bX\|^3<+\infty$.

\item[(C7')] Parameter $\rho=e^{\alpha_1}$ is a constant (does not
depend on $n$) and  $n^{-\frac{1}{2}}\ll\lambda\ll\lambda_0/\log
n\ll1$.
\item[(C8')]
$$\lambda_0<\lambda_{\min}
[\mathrm{E}(\bX\bX^T\frac{e^{\bX^T\boldsymbol\theta_0}}{(1+e^{
\bX^T\boldsymbol\theta_0})^2})](1+\rho)/\rho.$$
\end{enumerate}

The following proposition shows that the results of Theorem
\ref{theo:LAMPasymp} carry over to the penalized logistic
regression when (C6')--(C8') are assumed.

\begin{corollary}
For the penalized maximum likelihood estimator of logistic
regression with the sigmoid penalty \eqref{eq::sigmoid}, we have
parameter estimation consistency, model selection consistency and
strong asymptotic stability under conditions (C6')-(C8').
\end{corollary}
\section{Algorithms}

An important aspect of the penalized likelihood estimation method
is the computational efficiency. For the LASSO penalty,
\cite{leastangle} proposed the path-following LARS algorithm. In
\cite{friedman07}, the coordinate-wise descent method was
proposed. It optimizes a target function with respect to a single
parameter at a time, cycling through all parameters until
convergence is reached. For non-convex penalties, \cite{Fan01}
used the LQA approximation approach. In \cite{onestep}, a local
linear approximation type method was proposed for maximizing the
penalized likelihood for a broad class of penalty functions. In
\cite{FanLv11}, the coordinate-wise descent method was implemented
for non-convex penalties as well. \cite{YuFeng2012} proposed a
hybrid approach of Newton-Raphson and coordinate descent for
calculating the approximate path for penalized likelihood
estimators with both convex and non-convex penalties.

\subsection{Iteratively reweighted least squares}
We apply quadratic approximation and use the coordinate decent
algorithm similar to \cite{fasttracking} and \cite{viol}. Recall
that  our objective function is the negative penalized
log-likelihood $-\tilde
l({\btheta})=-l(\btheta)+n\sum_{j=1}^pp_{\lambda}(|\beta_j|).$
Following  \cite{fasttracking}, let $F_j=\frac{\partial
l(\btheta)}{\partial\theta_j},$ and define the violation function
$$\mathrm{viol}_j(\btheta)\triangleq\left\{
                       \begin{array}{ll}
                         |F_j|, & \mbox{if } j=0, \\
                         \max\{0,-n\lambda-F_j,-n\lambda+F_j\}, & \hbox{if $\theta_j
                           =0, \; j>0$,} \\
                         |F_j-n\mathrm{sgn}(\theta_j)p_{\lambda}'(|\theta_j|)|, & \hbox{if $\theta_j\neq0, \;j>0$.}
                       \end{array}
                     \right.$$

We see that the objective function achieves its maximum value if and
only if $\mathrm{viol}_j=0,$ for all $j$. Thus we use $\max_j\{\mathrm{viol}_j\}<\tau$
as the stop condition of our iteration, with  $\tau>0$ being the chosen tolerance threshold.

In each step, we use quadratic approximation to the log-likelihood
function $-l(\btheta)\approx\frac{1}{2}
(\tilde{\bY}-\mX\btheta)^T\mW(\tilde{\bY}-\mX\btheta)+\mathrm{constant},$
where $\mW$ and $\tilde{\bY}$ depend on the current value of
$\btheta$. The algorithm is summarized as follows.

 \textbf{Algorithm:} Set values for $\tau>0, \lambda,\lambda_0,\rho>0$. Denote by $\bX_{\cdot
j}$ the $(j+1)$th column of $\mX,j=0,\cdots,p.$
\begin{itemize}
  \item [S1.] Standardize $\bX_i,i=1,2,\cdots,n.$
  \item [S2.] Initialize $\btheta=\btheta^{(0)}$.
  Calculate $\mathrm{viol}=\max\{\mathrm{viol}_j(\btheta)\}$ and go
to S3 if $\mathrm{viol}>\tau$ or else go to S5.
  \item [S3.] Choose $j^*\in \arg\max_j \mathrm{viol}_{j}(\btheta)$.
Calculate $\mW,v=\frac{1}{n}\bX'_{\cdot j^*}\mW\bX_{\cdot j^*},$
$z=\frac{1}{n}\bX'_{\cdot j^*}$
$\mW(\tilde{\bY}-\mX\btheta)+v\theta_{j^*}.$ If $j^*\neq 0,$ let
$r=p'_{\lambda}(|\theta_j^*|)$; else $r=0$.
  \item [S4.] If $j^*=0, \theta_{j^*}=z/v;$ else
 $\theta_{j^*}=\mbox{sign}(z)(|z|-r)_+$.

Calculate $\mathrm{viol}=\max_j\{\mathrm{viol_j}\}$ and go to
S3 if $\mathrm{viol}>\tau$,  else go to S5.
  \item [S5.] Do the transformation of the coefficients $\theta$ due to
standardization.
\end{itemize}

For S2, the initial solution $\btheta^{(0)}$ can be taken as the zero solution or the MLE or the estimate calculated using a parameter $\lambda^*\in U(\lambda;\epsilon)$ from the previous steps.

For S4, we first carry out the iterations for the variables in the current active set until convergence, then check whether additional variables should join the active set. Alternatively, we may speed up the calculation  by using ``warm start". Readers are referred to \cite{fasttracking}
and \cite{glmnetpackage} for details of the strategies to speed up calculation in coordinate descent algorithms.

For example, the logistic regression with sigmoid penalty is,
$\sum_{i=1}^n$
$\log(1+e^{-Y_i\bX^T_i\theta})+n\sum_{j=1}^pp_{\lambda}(|\beta_j|),$
where $p_{\lambda}(\theta)=
\frac{\lambda^2_n(1+\rho)}{\lambda_0}\log[(1+\rho)e^{\lambda_0/\lambda\theta}/(1+\rho
e^{\lambda_0/\lambda\theta})],$ for $\theta>0.$ We define
$r_i=\exp(-Y_i\bX^T_i\btheta),i=1,\cdots,n$, and
$F_j(\theta)=\sum_{i=1}^nr_iY_i X_{i,j}/(1+r_i),j=0,\cdots,d$.  In
S3, we have the following two different approximation methods for
updating.
\begin{enumerate}
  \item Quadratic approximation from IRLS \citep{IRLS06}. Let $\mW=$
  $\frac{1}{2}$ $\mathrm{diag}\{\tanh(\frac{\pi_1}{2})$ $/
\pi_1,\cdots,\tanh(\frac{\pi_n}{2})/\pi_n\}$ and
$\tilde{\bY}=\frac{1}{2}\mW^{-1}\bY,$ where $
\pi_i=\bX^T_i\btheta^{(0)} Y_i$.

  \item Quadratic approximation using Taylor expansion.
  Let $\mW=\mathrm{diag}\{\pi_1(1-\pi_1),\cdots,\pi_n(1-\pi_n)\}$ and $\tilde{\bY}=\bX\btheta^{(0)}+\mW^{-1}[\bY\circ(1-\bpi)],$
where $\pi_i=1/[1+\exp(-Y_i\bX_i^T\btheta^{(0)})]$ and $\circ$ is
the component-wise product operator.
\end{enumerate}

For an initial estimator $\btheta^{(0)}$, denote by $a(\btheta,\btheta^{(0)})$ a quadratic approximation of $-n^{-1}l(\btheta)$ at $\btheta^{(0)}$, i.e.,
$$a(\btheta^{(0)},\btheta^{(0)})=-n^{-1}l(\btheta^{(0)}),~\frac{\partial{
a(\btheta,\btheta^{(0)})}}{
\partial{\btheta}}|_{\btheta=\btheta^{(0)}}=-n^{-1}l'(\btheta^{(0)}).$$
\begin{theorem}
Let $\mathcal{C}\subset R^d$ be a closed set and the objective
function
$M_n(\btheta)=-n^{-1}l(\btheta)+n\sum_{j=1}^pp_{\lambda}(|\theta_j|)$
is strictly convex. In addition, assume the quadratic
approximation at $\btheta^{(0)}$ satisfies
$a(\btheta,\btheta^{(0)})\geq -n^{-1}l(\btheta)$ for all
$\btheta\in \mathcal{C}$. Then the algorithm constrained in
$\mathcal{C}$ (minimization within $\mathcal{C}$) converges
to the true minimum
$\theta^*=\mathrm{arg}\min\limits_{\btheta}M_n(\btheta)$. In
addition, method 1 for the logistic regression satisfies the
conditions on quadratic approximation.
\label{theo:algorithmconverge}
\end{theorem}


\subsection{Balance between stability and parsimony}
\label{sec:alg:bal} An important issue is how to choose tuning
parameters in the penalized likelihood estimation.  For the LAMP
family, there are three tuning parameters, i.e., $\lambda$,
$\lambda_0$ and $\alpha_1$. Our numerical experiences show that
the resulting estimator is not sensitive to the choice of
$\alpha_1$. In most cases, we may simply take $\alpha_1=-1$ or
$\alpha_1=0$, depending on the type of regression (See Remark
\ref{remark::alpha_1_choice}).  For $\lambda$ and $\lambda_0$, we
recommend using cross-validation or BIC, so long as the solutions
are stable enough. The \pkg{ncvreg} package described in
\cite{coordinate11} to determine a stable area or perform local
diagnosis is recommended. There are two approaches to get a stable
area: to control the smoothness of the $\lambda$-estimate curve
and calculate the smallest eigenvalue of the penalized likelihood
at each point of the path as stated in Theorem 1. Here, we take
the second approach in all numerical analysis.

 Our algorithm differs from \pkg{ncvreg} in the following two aspects:
we use the ``viol" function as the convergence criteria;
we do not use the adaptive-scale. 
Both algorithms 1 and 2 use the linear approximation (suppose at
$\theta^{(0)}$) of the penalty term. $p_{\lambda}(|\theta|)\approx
p_{\lambda}(|\theta^{(0)}|)+p'_{\lambda}(|\theta^{(0)}|)(|\theta|-|\theta_0|).$
For algorithm 1, from concavity,  we have $p_{\lambda}(|\theta|)<
p_{\lambda}(|\theta^{(0)}|)+p'_{\lambda}(|\theta^{(0)}|)(|\theta|-|\theta_0|),$
 which naturally falls into the MM-algorithm framework.

To choose the $(\lambda_0, \lambda)$ pair, we use the  hybrid approach introduced by
\cite{coordinate11}, i.e.,  combining BIC, cross-validation, and
convexity diagnostics. For a path of solutions with a
given value of $\lambda_0$ large enough, use AIC/BIC to select
$\lambda$ and use the convexity diagnostics to determine the locally
convex regions of the solution path. If the chosen solution lies
outside the stable region, we lower $\lambda_0$ to make the
penalty more convex. After this process is iterated multiple times, we can
 find a value of $\lambda_0$
 that produces a good balance between sparsity and
convexity. Then we can fix $\lambda_0$ and use BIC or cross-validation to choose the best $\lambda$.
\section{Simulation results and examples}\label{sec:simulation}
Simulation studies cover logistic, Poisson and probit cases. The
performance of the LAMP family is compared with those of LASSO, SCAD
and MCP. Particular attention is given to the logistic regression to
demonstrate how sparsity and stability are properly balanced. Two
classification examples from microarray experiments involving cancer
patients are presented.
\subsection{Logistic regression}\label{sec:simu:logisticre}

We simulate from the logistic regression model with $n=200$,
$p=1000$, $\alpha=0$, $\bbeta=(1.5,1,-0.7,\boldsymbol
0_{997}^T)^T$, and $\bX\sim N(\boldsymbol 0, \bSigma)$, where
$\Sigma_{i,j}=\rho + (1-\rho)1_{i= j}$ with $\rho=0.5$. The number
of replications is 100 for this and all the subsequent
simulations.

Table 1 reports true positive (TP), false positve (FP),
proportion of correct fit (cf), proportion of over fit (of),
proportion of under fit (uf), $|\hat\bbeta-\bbeta|_1$ (L1 loss), and
$\|\hat\bbeta-\bbeta\|^2$ (L2 loss). To compare performances among LASSO, SCAD, MCP and the sigmoid
penalty, we use \pkg{glmnet} to calculate the LASSO solution path, \pkg{ncvreg} to calculate the SCAD and MCP. For all penalties, we use EBIC \citep{ebic} to choose the tuning
parameter $\lambda$ with other parameters fixed. The EBIC
parameter $\eta=1$.

From Table \ref{tab:logisticcv}, it is clear that the sigmoid penalty outperforms SCAD and MCP in the sense that at a similar level of TP, the sigmoid penalty results in a smaller FP. Somewhat surprisingly the LASSO has a competitive performance, which may be attributed to the use of the EBIC selection criterion.

\begin{table}\caption{ Simulation results for model selection and estimation under different penalties. Entries are mean values  with the standard error in parentheses. Here ``sig" represents the sigmoid penalty. The number inside the parentheses following by the penalty name represents the concavity parameter. \label{tab:logisticcv}
}  \begin{center}
\begin{tabular}{llllllll}
  \hline
  Penalties&TP&FP&cf&of&uf&L1&L2\\\hline
LASSO&1.87 (0.053)&0.12 (0.041)&0.05&0.01&0.94&2.49&1.46\\
sig(.02)&1.87 (0.053)&0.11 (0.040)&0.05&0.01&0.94&2.48&1.45\\
sig(.03)&1.90 (0.052)&0.13 (0.042)&0.06&0.01&0.93&2.41&1.41\\
sig(.05)&1.95 (0.05)&0.15 (0.041)&0.08&0.01&0.91&2.34&1.36\\
sig(.07)&1.99 (0.048)&0.20 (0.047)&0.09&0.02&0.89&2.25&1.31\\
sig(.09)&2.02 (0.047)&0.20 (0.049)&0.1&0.02&0.88&2.17&1.26\\
sig(.15)&2.21 (0.057)&1.86 (0.43)&0.1&0.19&0.71&3.18&1.36\\
sig(.38)&2.61 (0.049)&4.65 (0.40)&0.1&0.51&0.39&6.23&2.13\\
SCAD(300)&1.69 (0.073)&0.14 (0.043)&0.04&0.01&0.95&2.94&1.74\\
SCAD(7)&1.69 (0.073)&0.14 (0.043)&0.04&0.01&0.95&2.94&1.74\\
SCAD(5)&1.84 (0.077)&3.20 (0.69)&0.03&0.11&0.86&2.94&1.74\\
SCAD(4)&2.25 (0.089)&12 (0.75)&0.01&0.45&0.54&91.9&23.5\\
MCP(300)&1.69 (0.073)&0.14 (0.043)&0.04&0.01&0.95&2.94&1.74\\
MCP(50)&1.70 (0.073)&0.16 (0.047)&0.04&0.01&0.95&2.94&1.74\\
MCP(15)&1.75 (0.073)&0.14 (0.043)&0.06&0.01&0.93&2.45&1.45\\
MCP(7)&1.79 (0.071)&0.14 (0.040)&0.07&0.01&0.92&2.21&1.31\\
MCP(5)&1.91 (0.081)&2.98 (0.69)&0.09&0.1&0.81&2.07&1.23\\
MCP(4)&2.29 (0.092)&11.89 (0.80)&0.02&0.50&0.48&109.3&28.7\\
\end{tabular}
\end{center}

\end{table}

\subsection{Smoothness}\label{sec:simu:shapeeffect}
\begin{table}
\caption{Simulation results for TP and FP under different penalties with the tuning parameter selected by BIC. Mean values are presented. \label{tab:logisticbic}}
\begin{center}

\begin{tabular}{lll|lll|lll}
  \hline
  \multicolumn{3}{c|}{sigmoid}&\multicolumn{3}{c|}{SCAD}&\multicolumn{3}{c}{MCP} \\ \hline
  $\lambda_0$&TP&FP&$a$&TP&FP&$\gamma$&TP&FP\\\hline
   .04 & 1.86 & 0.28 & 1.1&1.89&0.21& 1.1 & 1.90 & 0.21 \\
   .05 & 1.81 & 0.22  & 7 & 1.89&0.21& 7 &1.90& 0.21\\
   .06 & 1.84 & 0.24 & 14 & 1.89& 0.21& 14 & 1.90& 0.21\\
   .08 & 1.89 & 0.22 &20& 1.89& 0.21& 20& 1.90& 0.21\\
   .10 & 1.91 & 0.22 & 27& 1.88& 0.23& 27& 1.89& 0.22\\
   .11 & 1.89 & 0.21 & 34 & 1.83 & 0.23& 34 & 1.86 & 0.24 \\
   .13 & 1.89 & 0.21 & 41 & 1.83 & 0.24& 41 & 1.83 & 0.23 \\
   .15 & 1.89 & 0.21 & 54 & 1.83 & 0.28& 54 & 1.83 & 0.29\\
   .16 & 1.89 & 0.21 &60 & 1.86 & 0.32 &60 & 1.86 & 0.32\\
   .18 & 1.89 & 0.21&67 & 1.88 & 0.35 &67 & 1.88 & 0.36 \\
\end{tabular}

\end{center}

\end{table}

\begin{figure}[!t]\caption{Solution paths  under different penalties. \label{fig:solutionpathsimu}}
\begin{center}

\subfigure[sigmoid(0.1)]{\label{fig:subfig:i}\resizebox{2in}{2in}{\includegraphics{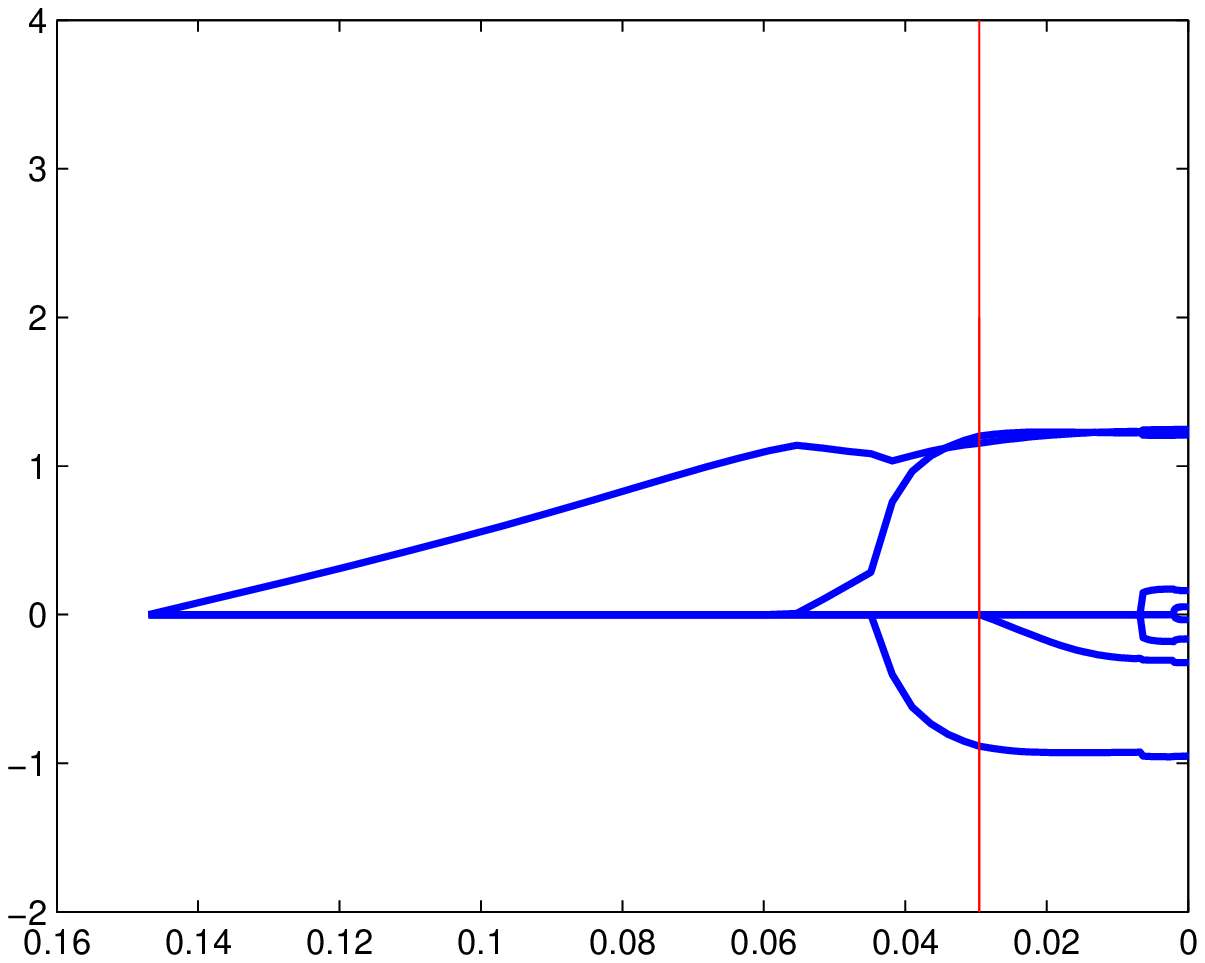}}}
\subfigure[LASSO]{\label{fig:subfig:i}\resizebox{2in}{2in}{\includegraphics{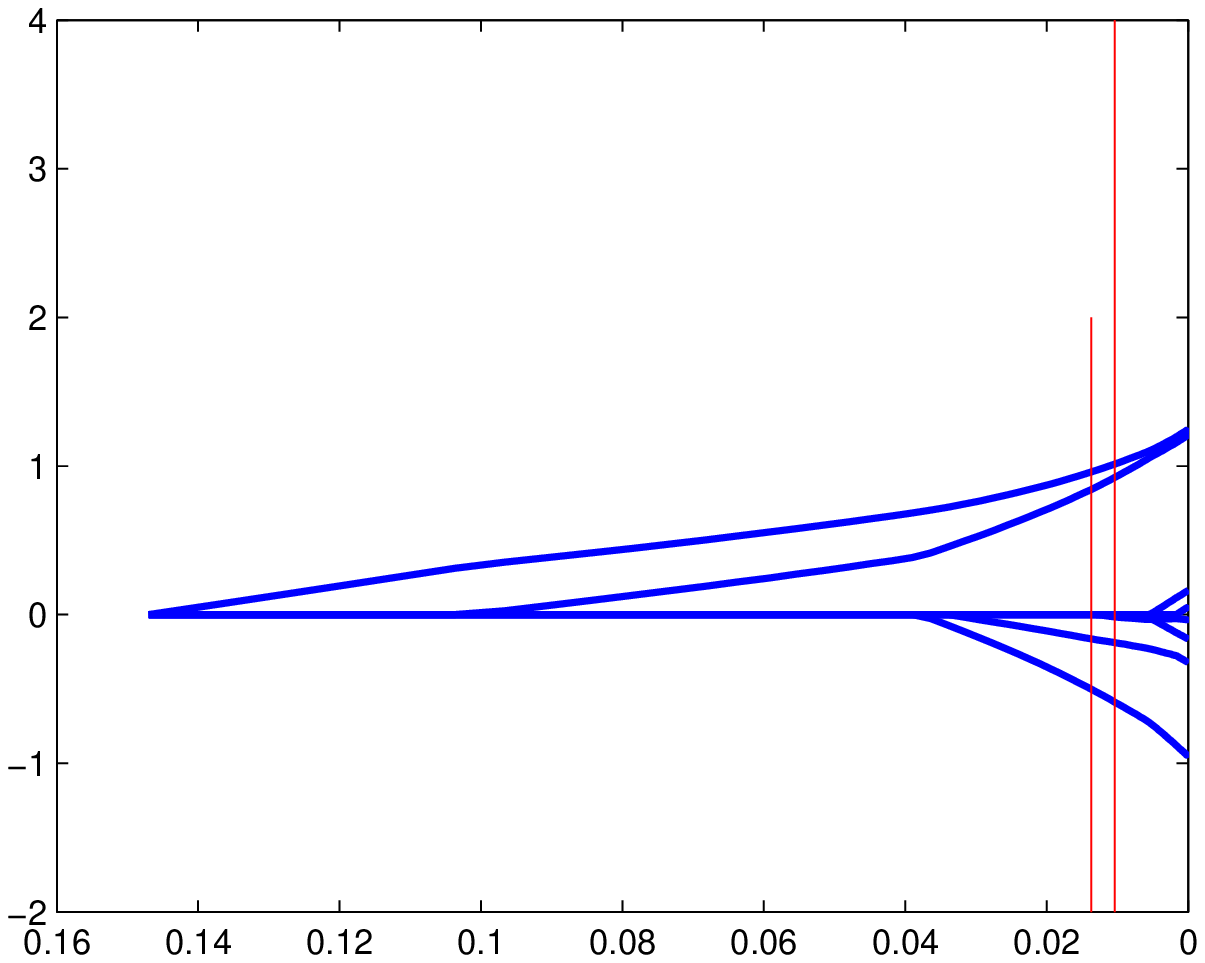}}}
\subfigure[SCAD(20)]{\label{fig:subfig:a}\resizebox{2in}{2in}{\includegraphics{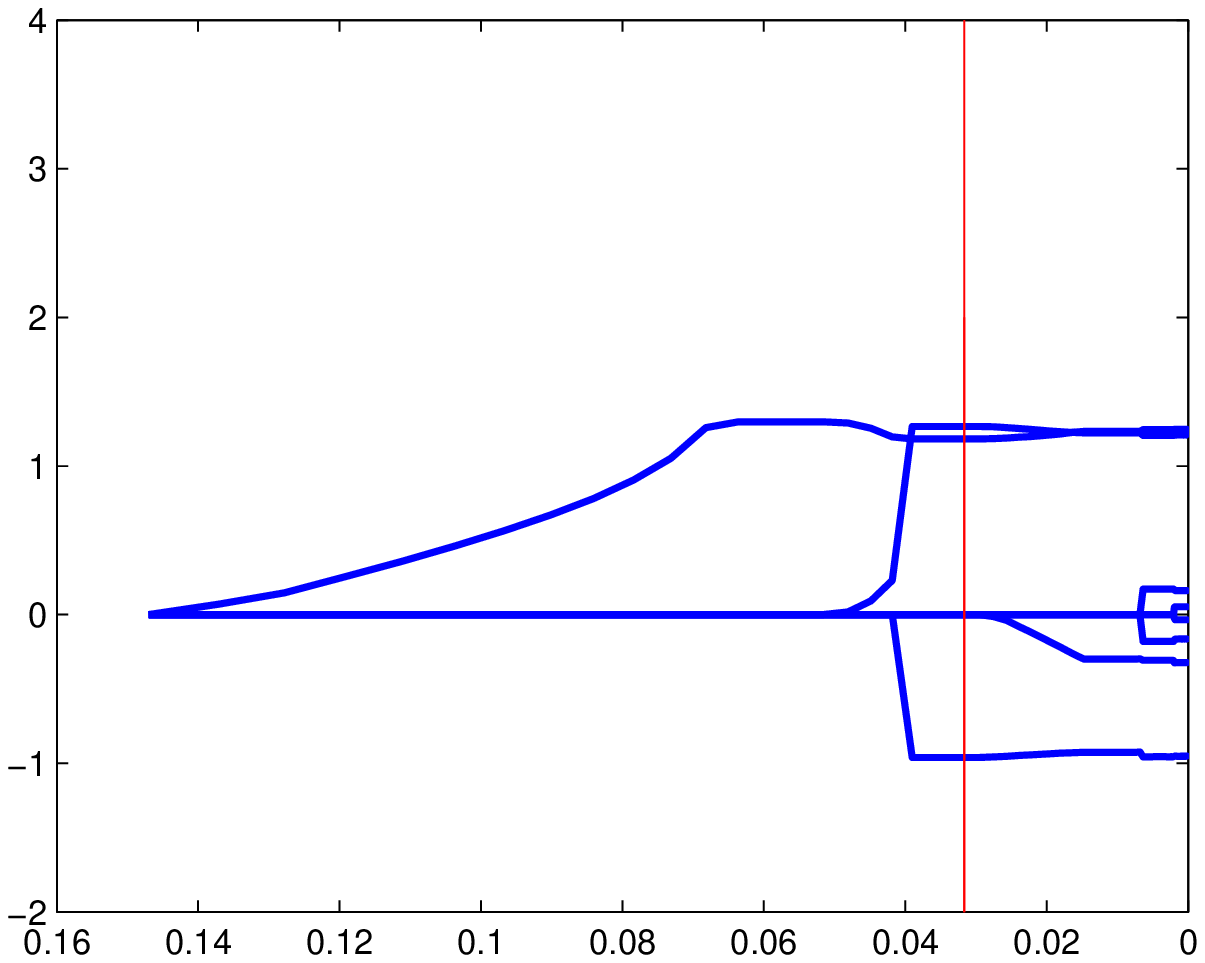}}}
\subfigure[MCP(20)]{\label{fig:subfig:a}\resizebox{2in}{2in}{\includegraphics{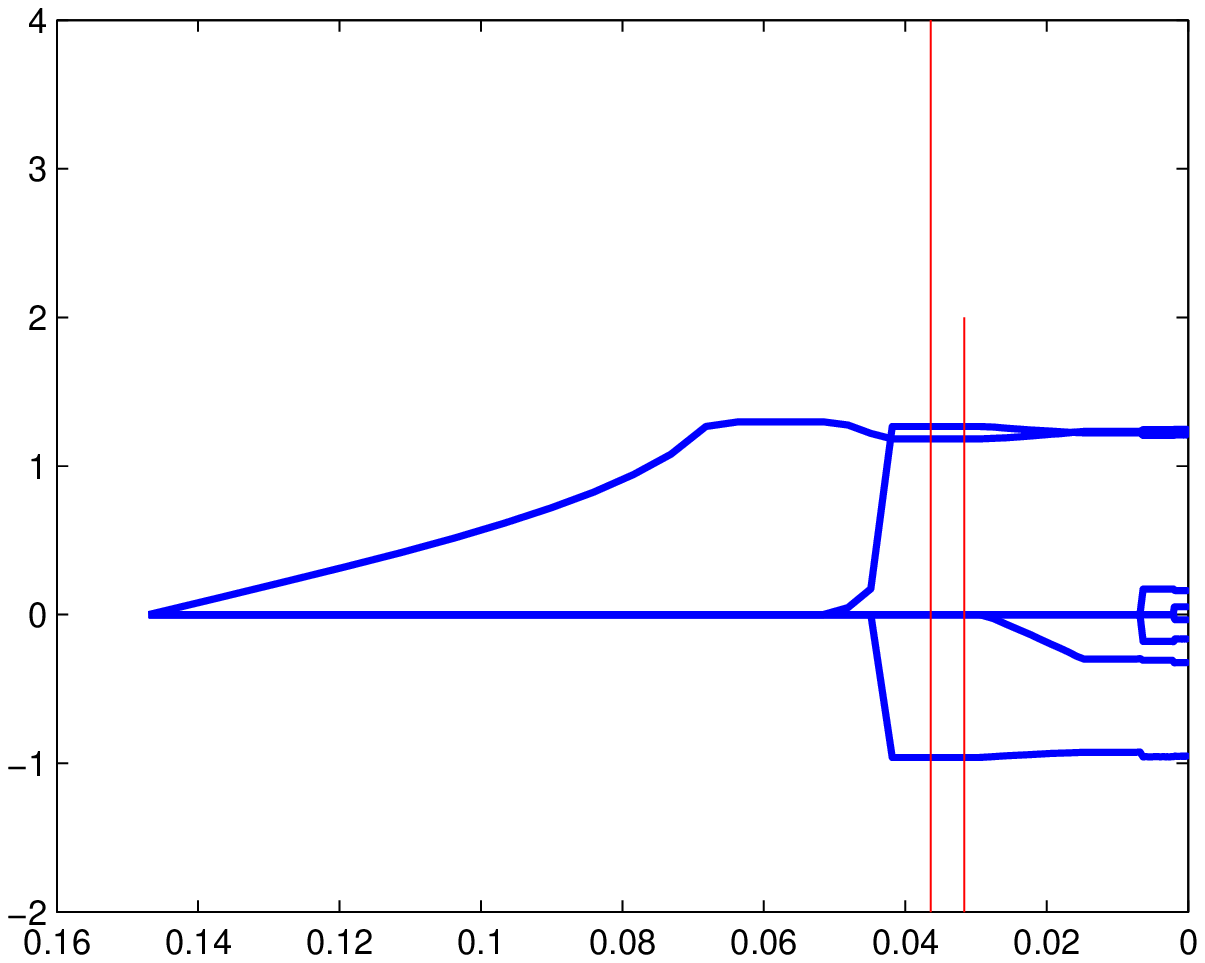}}}

\end{center}
\end{figure}

The same logistic regression model as in Section
\ref{sec:simu:logisticre} except $\alpha_0=-3$, $\boldsymbol\beta_0=(1.5,1,-0.7,0,0,0,0,0)^T$ and $\Sigma_{i,j}=0.5^{|i-j|}$  is used.  In Figure
\ref{fig:solutionpathsimu}, we compare smoothness of solution paths generated from the sigmoid
penalty, SCAD and MCP with the same concavity at 0. It can be seen that this choice will also lead to similar sparse level as in Table
\ref{tab:logisticbic}. SCAD and MCP use the same algorithm as sigmoid does
(not adaptive-scale as in \pkg{ncvreg} package). The shorter vertical line is
the BIC choice of $\lambda$ and the longer one uses a 10-fold cross
validation. To avoid variation due to the random division in the
cross-validation, the result in Table \ref{tab:logisticbic} uses BIC
to choose $\lambda$ with the other tuning parameter fixed. The subfigures
(c) and (d)  for SCAD and MCP are both less smooth than
subfigure (a) for the sigmoid penalty which is of similar smoothness as subfigure (b) for LASSO. The sigmoid penalty appears to outperform SCAD and MCP in terms of smoothness of the solution path.

\subsection{Stability}\label{sec:simu:stability} The data are generated in the same way as in the preceding subsection.
For each replication, we repeat
100 times cross-validation to select $\lambda$ and calculate its mean sample standard deviation. The box-plots for the mean sample standard deviations are generated. In addition, to evaluate the asymptotic stability as introduced in Section \ref{sec:simu:stability}, we add a small random perturbation generated from $\mathrm{N}(0,0.1)$ to
all the observations before conducting the analysis.
The box-plot results are presented in Figure \ref{fig:boxplot}.

The result without the random error term evaluates the stability regarding the randomness of cross-validation for each penalty. The result with the random error evaluates the stability towards the random perturbation of the data.
To ensure a fair
comparison, we choose the same level of concavity at 0 for SCAD, MCP and the sigmoid penalty. It is seen from the box-plots
 that LASSO is the most stable one, while the sigmoid penalty outperforms SCAD and
 MCP, in terms of both the median and 75\% quantile in the case without error, and 75\%
 quantile in the case with the error added.

\begin{figure}[t]
\caption{Box-plots for the mean standard deviations in Section \ref{sec:simu:stability}. The left panel shows the box-plot without perturbation and the right one
is that with perturbation. \label{fig:boxplot}}
\begin{center}

\resizebox{2.1in}{2.3in}{\includegraphics{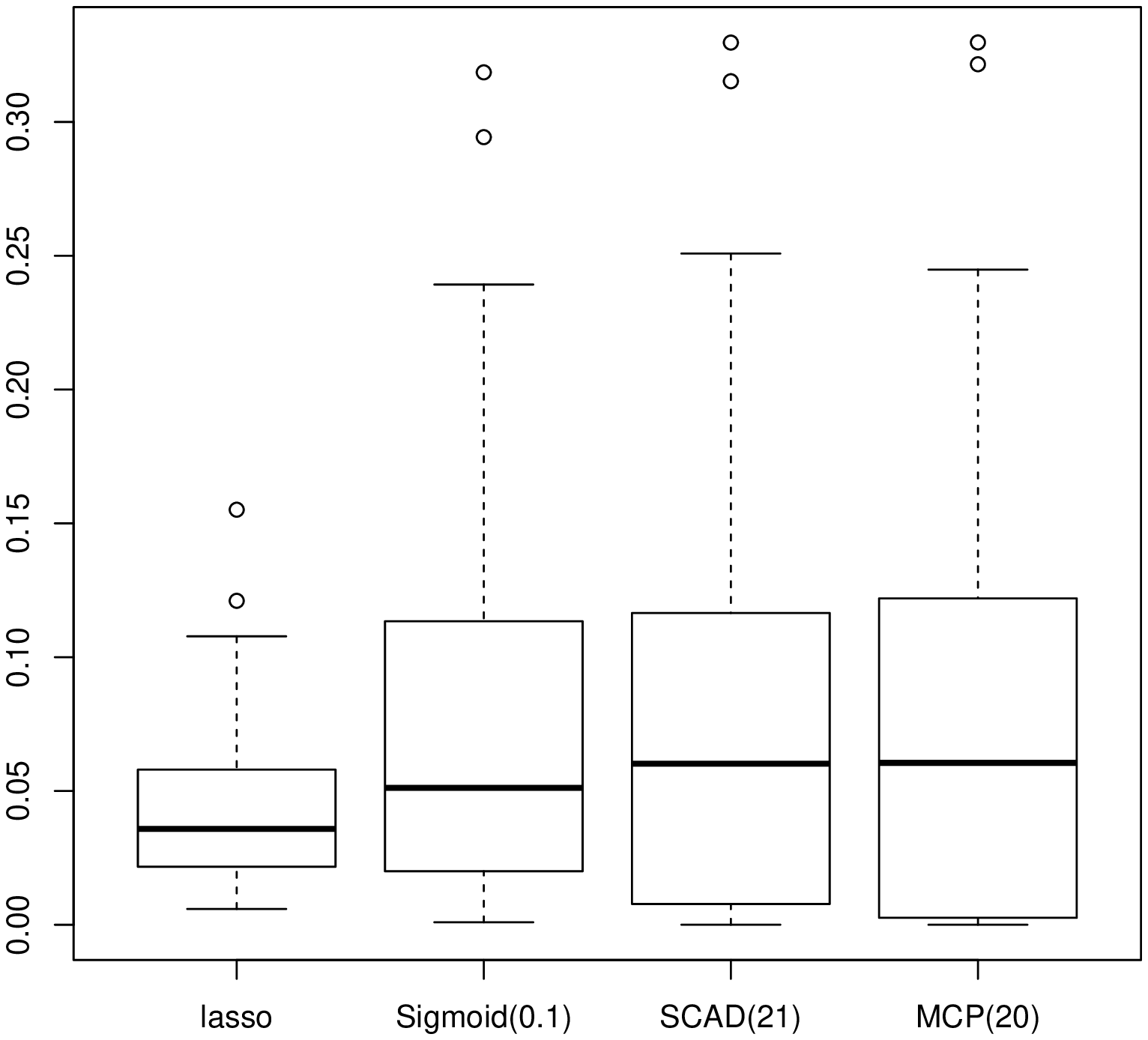}}
\resizebox{2.1in}{2.3in}{\includegraphics{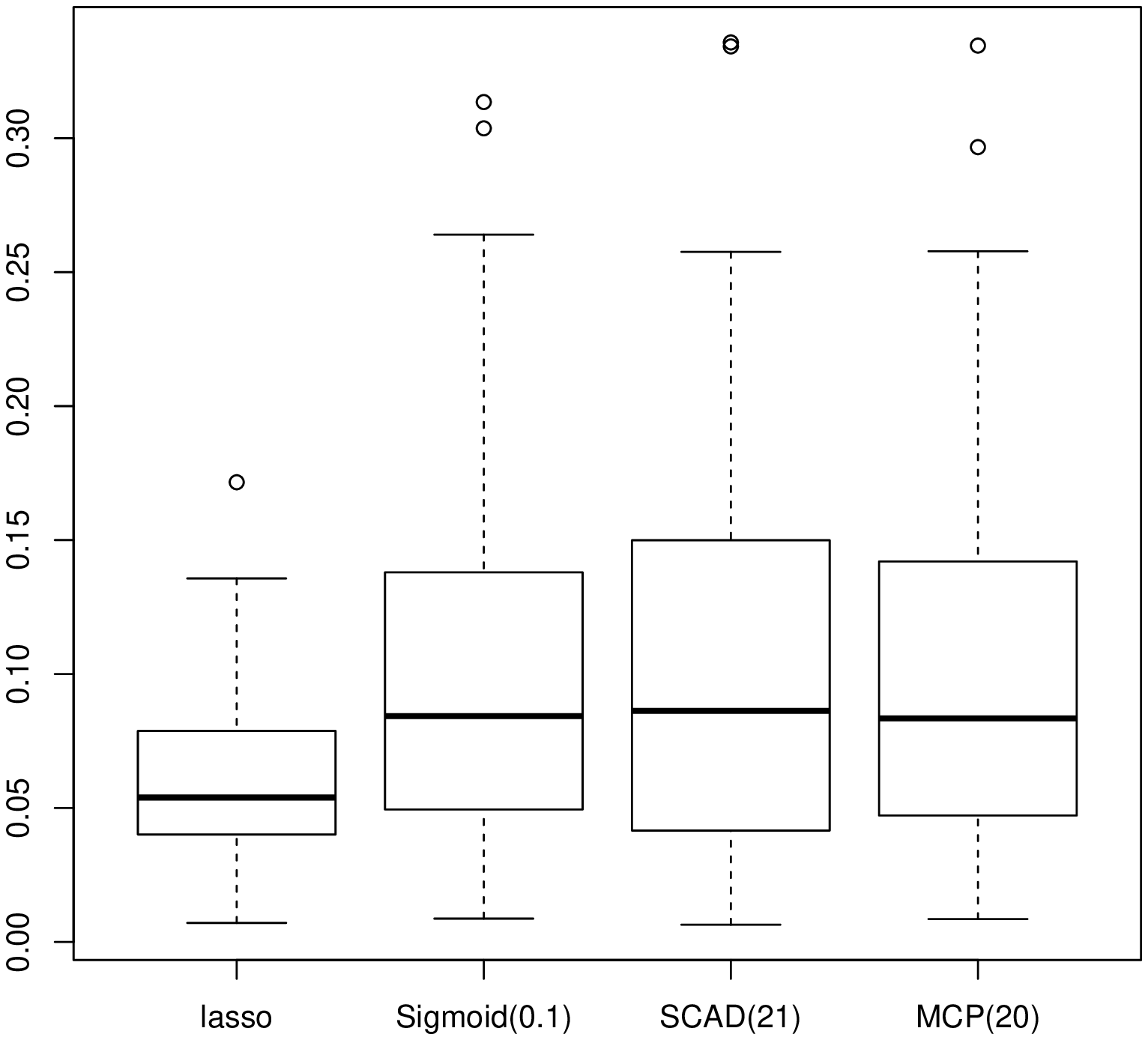}}

\end{center}
\end{figure}

\subsection{Poisson and probit} We simulate from the Poisson regression model with $n=250$, $\alpha=-1$,
 $\boldsymbol\beta=(.6,.4, 0, 0,1,0, {\bf 0}_9^T)^T$, and $\bX\sim N(\boldsymbol 0, \bSigma)$, where
$\Sigma_{i,j}=\rho + (1-\rho)1_{i= j}$ with $\rho=0.5$. For probit regression, we simulate the $\bx$ the same way and set $\alpha=-2$ and $\bbeta=(3,2,1,{\bf 0}_8^T)^T$.

The results for the Poisson and probit regression models are
summarized in Tables \ref{tab:poissonbic} and \ref{tab:probitbic},
respectively. For the Poisson regression, we also report the MRME
(median of ratios of model error of a selected model to that of
the ordinary maximum likelihood estimate under the full model). We
observe similar selection performance using Poisson, SCAD and MCP.
However, the Poisson penalty leads to the smallest MRME.  For the
probit model, all three nonconvex penalties lead to the same
result for all criteria considered.
\begin{table}\caption{Simulation results for Poisson regression under different penalties.\label{tab:poissonbic}}

\begin{center}
 \begin{tabular}{lllllllll}
  \hline
  Penalties&TP&FP&cf&of&uf&L1&L2&MRME\\\hline
LASSO&2.98 (0.01)&1.40 (0.13)&0.26&0.72&0.02&0.63&0.34&0.49\\
Poisson&2.80 (0.04)&0.41 (0.08)&0.65&0.15&0.20&0.49&0.30&0.19\\
SCAD&2.78 (0.04)&0.35 (0.07)&0.54&0.24&0.22&0.47&0.30&0.20\\
MCP&2.81 (0.04)&0.38 (0.08)&0.66&0.15&0.19&0.48&0.29&0.22\\
\end{tabular}
\end{center}
 \end{table}

\begin{table}\caption{Simulation results for probit regression under different penalties.\label{tab:probitbic}}
\begin{center}
 \begin{tabular}{lllllllll}
  \hline
  Penalties&TP&FP&cf&of&uf&L1&L2&MRME\\\hline
LASSO&3.00 (0)&2.13 (0.18)&0.21&0.79&0&2.47&1.35&0.71\\
probit&2.99 (0.01)&0.20 (0.05)&  0.82&0.17&0.01&1.37&0.85&0.27\\
SCAD&2.99 (0.01)&0.20 (0.05)&  0.82&0.17&0.01&1.37&0.85&0.27\\
MCP&2.99 (0.01)&0.20 (0.05)&  0.82&0.17&0.01&1.37&0.85&0.27\\
\end{tabular}
\end{center}

\end{table}

\subsection{Examples}
We apply the proposed LAMP to two gene expression datasets: lung cancer data (Gordon et al., 2002), and prostate cancer data (Singh et al., 2002). The two datasets are downloadable at \url{http://www.chestsurg.org}
and \url{http://www.broad.mit.edu}. The response variable in each dataset is binary.

We aim to use the lung cancer data to classify
malignant pleural mesothelioma (MPM) from adenocarcinoma (ADCA) of the
lung. The data consists of 181 tissue samples, 32 of which are
for training with remaining 149 for testing. Each
sample is described by 12533 genes.

After the initial standardization of the predictors into mean zero and variance
one,  we apply LASSO, SCAD, MCP, and the sigmoid penalty, using
\pkg{glmnet} for LASSO, \pkg{ncvreg} for SCAD and MCP. For each
method, a 10-fold cross-validation is used to select the best
$\lambda$. We repeat 10 times to make different divisions to
calculate the cross-validation error. For SCAD and MCP,  we evaluate
the performance when $a,\gamma \in [3,6]$ while for sigmoid penalty,
$\lambda_0\in (.005,0.03)$. The results are summarized in Table
\ref{tab:realdatalung}.
\begin{table}
\caption{Classifications errors and selected model sizes under different penalties for lung cancer data. \label{tab:realdatalung}}

\begin{center}
 \begin{tabular}{lllllllll}
  \hline
  Penalties & $\lambda_0/a/\gamma$ & Training error & Test error & \# of selected genes\\\hline
LASSO&\ \ ---&\ \ 0/32 &\ 7/149 & \ \ \ 15\\
SCAD&(3$\sim$6)&\ \ 0/32 &\ 6/149 &\ \ \ 13 \\
MCP&(3$\sim$6)&\ \ 0/32&\ 7/149 & \ \ \ \ 3\\
sigmoid&(.022$\sim$0.03) &\ \ 0/32 &\ 7/149 &\ \ \ \ 3\\
sigmoid&(.021)&\ \ 0/32 &\ 6/149 &\ \ \ \ 5\\
\end{tabular}
\end{center}
\end{table}
The result for sigmoid penalty is quite similar to MCP when $\lambda_0\in (.022, 0.03)$. When $\lambda_0=0.021$, the sigmoid have 6 test errors with only 5 genes selected, which is better compared with SCAD.

For the prostate cancer data, the goal is to classify prostate tumor samples from the normal samples. There are 102 patient samples for
training, and 34 patient samples for testing with 12600 genes
 in total. The result are reported in
Table \ref{tab:realdataprostate}. Here we see the test errors are similar across methods,  although the sigmoid penalty leads to the most sparse solution.

\begin{table}\caption{Classifications errors and selected model sizes under different penalties for prostate cancer data.\label{tab:realdataprostate}}
\begin{center}
 \begin{tabular}{lllllllll}
  \hline
  Penalties & $\lambda_0/a/\gamma$ & Training error & Test error & \# of selected genes\\\hline
LASSO&\ \ ---&\ \ 0/102 &\ 2/34 & \ \ \ \ 30 \\
SCAD&(20$\sim$25)&\ \ 0/102 &\ 2/34 & \ \ \ \ 26 \\
MCP&(35$\sim$50)&\ \ 0/102&\ 2/34 & \ \ \ \ 24 \\
sigmoid&(.001) &\ \ 0/102 &\ 1/34 & \ \ \ \ 26 \\
sigmoid&(.002)&\ \ 0/102 &\ 2/34 & \ \ \ \ 23 \\
\end{tabular}
\end{center}\end{table}

\section{Discussion}
Penalty based regularization methods have received much attention in recent years.  This paper proposes a family of penalty functions (LAMP) that is adaptive and specific to the shapes of the log-likelihood functions. The proposed LAMP family is different from the well-known LASSO, SCAD and MCP. It can be argued that the new approach provides a good balance between sparsity and stability, two important aspects in model selection. It is shown that the resulting penalized estimation achieves model selection consistency and strong asymptotic stability, in addition to the usual consistency and asymptotic normality.

An important issue is how to choose the three parameters imbedded in a LAMP. The ``location''
parameter $\alpha_1$ can be chosen in an obvious way for the standard generalized linear models,
 while $\lambda$, which represents the penalty level, can be chosen through standard CV or information criteria.
 For $\lambda_0$, which controls the concavity level, it is computationally intensive to use CV.
 It is desirable to develop more effective ways to select $\lambda_0$.
 It is also important to study  the stability of the solution path.

The LAMP approach can be modified to handle grouped variable selection. It will also be of interest to develop parallel results for semiparametric regression models.


%
%

\appendix
\section{General results}\label{appendix::general results}

We first state a general result about estimation consistency,
model selection consistency and asymptotic normality.

Consider the penalized log-likelihood function
$\tilde{l}(\btheta)$ as defined by \eqref{eq::lik-pen}.
Let $\alpha_0$ be the true value of $\alpha$. Recall
$\bbeta_{10}$ is the nonzero
part of
$\boldsymbol\beta_0$ and
$\bbeta_{20}={\bf 0}$. For notational simplicity,
let

Let $\bZ$, $\bZ_i$ be i.i.d. with
density $f(\cdot,\btheta)$ that satisfies the following regularity
conditions (see \cite{Fan01}):
\begin{enumerate}
  \item[(C10)] $\mathrm{E}_{\btheta}[\partial\log f(\bZ,\btheta)/\partial\theta_j]=0,$
$$\mathrm{E}_{\btheta}\left[\frac{\partial\log f(\bZ,\btheta)}{\partial\theta_j}
\frac{\partial\log f(\bZ, \btheta)}{\partial\theta_k}\right]=-
\mathrm{E}_{\btheta}[\frac{\partial^2\log f(\bZ,
\btheta)}{\partial\theta_j\partial\theta_k}],$$where
$j,k=0,1,2,\cdots,p$;
  \item[(C11)]
  $$0<\mathrm{E}\left[\frac{\partial\log f(\bZ,\btheta)}{\partial\btheta}|_{\btheta=\btheta_0}\right]^{\bigotimes 2}<\infty;$$
  \item[(C12)] There exist functions $M_{jkl}$, a neighborhood $\bTheta_1$ of $\btheta_0$ such that
   $\E_{\btheta_0}
  M_{jkl}(\bZ)<\infty$ and
  $$\left |\frac{\partial^3\log f(\bZ,\btheta)}{\partial\theta_j\partial\theta_k\partial\theta_l}\right |\leq
  M_{jkl}(\bZ)$$
  for all $\btheta\in \bTheta_1$, $j$, $k$ and $l$.
\end{enumerate}
Note that for the generalized linear models, (C10)-(C12) are satisfied under mild assumptions on covariates. 

$$p_{1,n}\triangleq\sup_{\beta\in[\zeta_1,\zeta_2],1\leq j\leq
q}|p'_{\lambda,j}(\beta)|; \;
p_{2,n}(\beta)\triangleq\inf_{q<j\leq p}p'_{\lambda,j}(\beta);\;
p_{2,n}\triangleq p_{2,n}(0)$$
$$p_{3,n}(\beta)\triangleq\sup_{x\in(0,\beta),q<j\leq
p}|p''_{\lambda,j}(x)|;\;
p_{4,n}\triangleq\sup_{\beta\in[\zeta_1,\zeta_2],1\leq j\leq
q}|p''_{\lambda,j}(\beta)|;$$
$$p_{5,n}\triangleq\inf_{\beta\in[\zeta_1,\zeta_2],1\leq j\leq
q}p''_{\lambda,j}(\beta);\; \bSigma_1=
\mathrm{diag}\{p''_{\lambda,1}(|\beta_{10}|),\cdots,
p''_{\lambda,q}(|\beta_{q0}|)\};$$
$$\bb=(0,p'_{\lambda,1}(|\beta_{10}|)\mbox{sgn}(\beta_{10}),
\cdots,p'_{\lambda,q}(|\beta_{q0}|)\mbox{sgn}(\beta_{q0}))^T.$$
Let $~\mR=
  \left(
    \begin{array}{ccc}
      \mR_{00} &\mR_{01} & \mR_{02} \\
      \mR_{10} &\mR_{11} & \mR_{12} \\
      \mR_{20} &\mR_{21} & \mR_{22} \\
    \end{array}
  \right)
  $ be the Fisher information matrix at $\btheta_0$ partitioned by
  the intercept, the nonzero and zero parts and
  $$\bSigma_2\triangleq\left(
            \begin{array}{cc}
             \mR_{00}  & \mR_{01} \\
             \mR_{10}  & \mR_{11}+\bSigma_1 \\
            \end{array}
          \right).$$
          For $i, j =1,2$, define
  $\mR^*_{ij}\triangleq\mR_{ij}-\mR_{i0}\mR_{00}^{-1}\mR_{0j}.$
Let $\|{\bf A}_{m\times n}\|_{\infty}=\max_{1\leq i\leq m}\sum_{j=1}^n|A_{ij}|$ be the $\infty$ norm of a matrix ${\bf A}$.

\begin{lemma}Suppose that both $-l(\btheta)$ and
$-\mathrm{E}l(\boldsymbol\theta)$ are strictly convex, satisfying (C10)-(C12), and that
$p_{\lambda,j},j=1,\cdots,p$ are continuous at 0. Assume there exist $
\epsilon_n^{(k)}<\zeta_1,k=1,2$ such that $n^{-\frac{1}{2}}\vee
p_{1,n}=o(\epsilon_n^{(1)})$, the k-th derivative of
$p_{\lambda,j}(\beta)$ is continuous for
$\theta\in(0,\epsilon_n^{(k)})\bigcup[\zeta_1,\zeta_2]$ and exists
at the limit $\theta=0_+$,
$$p'_{\lambda,j}(0_+),p''_{\lambda,j}(0_+)\in
\bar{\mathbb{R}}=[-\infty,+\infty].$$ $p'_{\lambda,j}(0)\triangleq
p'_{\lambda,j}(0_+), p''_{\lambda,j}(0)\triangleq
p''_{\lambda,j}(0_+).$ For $j=1,\dots, p$,
$$p_{\lambda,j}(0)=0; \;\; p_{\lambda,j}(\beta)\geq0,\; \beta>0.$$
Then, we have the following results.
\begin{enumerate}
  \item (Parameter estimation consistency) If
either of the following two conditions, (1.a) and (1.b), holds, then there exists a
consistent local maximizer $\hat{\btheta}_n.$
\begin{itemize}
  \item[(1.a)] $p_{1,n}+
  p_{4,n}\rightarrow0$ as $n\rightarrow \infty$.
  \item[(1.b)] $p_{1,n}\rightarrow0$ and $
  \varliminf_np_{5,n}>0.$
\end{itemize}
  \item (Model selection consistency)
If any of the following three  conditions, (2.a)-(2.c), holds,  then any
consistent local maximizer  is model
selection consistent.
\begin{itemize}
  \item[(2.a)] As $n\to \infty$, $\min(\sqrt{n}p_{2,n}
  ,p_{2,n}/p_{1,n})\rightarrow\infty$
 and for any ${u_n}=O(p_{1,n}+\frac{1}{\sqrt{n}})$,
  $ p_{2,n}=O(p_{2,n}(u_n))$.
  \item[(2.b)]  There exists $C>0$ such that ${p_{2,n}}/{p_{1,n}}\rightarrow C$ and $\|\mR^*_{21}\mR_{11}^{*-1}\|_{\infty}
  <C$;  $
   \sqrt{n}p_{2,n}
  \rightarrow\infty$; and
  $p_{3,n}(\epsilon_n^{(2)})+p_{4,n}\rightarrow0.$
\item[(2.c)] For any $u_n=O(p_{1,n}+1/\sqrt{n})$, $\min(\sqrt{n}p_{2,n}(u_n)
  ,{p_{2,n}(u_n)}/{p_{1,n}})\rightarrow\infty.$
\end{itemize}
  \item (Asymptotic normality)
 Assume the estimator has the model selection consistency as stated in 2. If $\sqrt{n}p_{1,n}\rightarrow0$,
 we have the asymptotic normality for the nonzero part,
\begin{eqnarray*}
   \sqrt{n}\bSigma_2
\{\left(
                                  \begin{array}{c}
                                    \hat\alpha-\alpha_{10} \\
                                    \hat{\bbeta}_1-\bbeta_{10} \\
                                  \end{array}
                                \right)
   +\bSigma_2^{-1}\bb\}
   \longrightarrow N(0,\left(
                                                      \begin{array}{cc}
                                                        \mR_{00} & \mR_{01} \\
                                                        \mR_{10} & \mR_{11} \\
                                                      \end{array}
                                                    \right)).
\end{eqnarray*}

\end{enumerate}
\label{lemma:oracle}
\end{lemma}
Lemma \ref{lemma:oracle} covers commonly used penalties
including LASSO, SCAD, MCP, adaptive LASSO, hard thresholding, bridge
and the LAMP family. Note that the conditions are satisfied for SCAD, and are slightly  weaker than those imposed in \cite{Fan01}. It is also easy to verify that condition (2.a) is satisfied for SCAD and MCP, (2.b) for LASSO and (2.c) for bridge.


Suppose contrary to (2.a), $\varlimsup_n\sqrt{n}p_{2,n}<\infty$ and  $p_{3,n}+p_{4,n}\rightarrow0.$ Then, the resulting estimator $\hat\bbeta_n$ does not have the model selection
consistency.

\section{Proofs}
\begin{proof}[Proof of Lemma \ref{lemma::stability:weak}]

 For any $\epsilon>0,$ we have
 $\|\mE_i\|/\sqrt{n}<\epsilon$, $i=1,2$. Select
$$
\bu_{i,n}\in\arglmin_{\btheta}M_{n}(\mZ_n+\mE_i,\btheta),i=1,2.$$
Define function $f_i(\bu)=m_n(\mZ_n+\mE_i,\bu),i=1,2.$

Throughout the proof, we use the following conventions: $o(1)$ denotes a quantity approaching  0 as  $n\rightarrow\infty$ and $\epsilon\rightarrow0$
simultaneously; $o(1,n)$ for  approaching 0 as
$n\rightarrow\infty$; $o(1,\epsilon;n)$ for approaching 0 as
$\epsilon\rightarrow0$ with $n$ fixed, and $O(1,n)$ for a bounded sequence with
$n\rightarrow\infty.$ If a subscript $p$ is used, then the convergence (boundedness) is in probability.

 First, from Condition (C2),
 \begin{eqnarray}\label{lemma2proof:2}
  \nonumber
      &~&\sup_{\bu\in\bTheta}|f_2(\bu)-f_1(\bu)|  \\\nonumber
      &\leq & \frac{1}{n}\sum_{i=1}^nK(\bZ_i)\|\bepsilon_{1i}-\bepsilon_{2i}\|\\\nonumber
      &\leq &\frac{1}{\sqrt{n}}\sqrt{\frac{1}{n}\sum_{i=1}^nK^2(\bZ_i)}\|\mE_1-\mE_2\|\\\nonumber
        &=&O_p(1,n)o(1,\epsilon;n).
 \end{eqnarray}
Similarly we have
$$\sup_{\bu\in\bTheta}|m_n(\mZ_n,\bu)-f_1(\bu)|=O_p(1,n)o(1,\epsilon;n).$$

Second, from Condition (C1), the $f_i(\bu)$ are convex.
 We will show that for any $r \in[0,1]$, $\bu'_{2,n}=r\bu_{1,n}+(1-r)\bu_{2,n}$, and $i=1,2$,
 \begin{equation}\label{lemma2proof:1}
    f_i(\bu'_{2,n})\geq
    f_i(\bu_{i,n})+o_p(1,n).
\end{equation}
 Since $\bu_{1,n}$ is a
local minimum of $f_1(\bu)+r_n(\bu)$, there exists
$\sigma_{1,n}\in (0, 1)$ such that
\begin{eqnarray}
\nonumber
  &\ &f_1(\bu_{1,n})+r_n(\bu_{1,n})-r_n(\bu_{1,n}+\sigma_{1,n}(\bu'_{2,n}-\bu_{1,n})) \\\nonumber
  &\leq&
  f_1(\bu_{1,n}+\sigma_{1,n}(\bu'_{2,n}-\bu_{1,n}))
  \\\nonumber
  &=& f_1(\sigma_{1,n}\bu'_{2,n}+(1-\sigma_{1,n})\bu_{1,n}) \\\nonumber
  &\leq& \sigma_{1,n}
  f_1(\bu'_{2,n})+(1-\sigma_{1,n})f_1(\bu_{1,n}),
\end{eqnarray}
where the last inequality follows from the convexity of $f_1(\cdot)$.
Therefore,
\begin{eqnarray} \nonumber
  f_1(\bu'_{2,n})&\geq&f_1(\bu_{1,n})-\frac{r_n(\bu_{1,n}+\sigma_{1,n}(\bu'_{2,n}-\bu_{1,n}))-r_n(\bu_{1,n})}{\sigma_{1,n}
  \|\bu'_{2,n}-\bu_{1,n}\|}\|\bu'_{2,n}-\bu_{1,n}\|,
\end{eqnarray}
which, in view of Condition (C3), implies \eqref{lemma2proof:1} for $i=1$. The case of $i=2$ can be proved similarly.

Third, $\mathrm{E}m_n(\mZ_n,\bu_{2,n}')-
m_n(\mZ_n,\bu'_{2,n})=o_p(1,n)$ from the law of large number. Note
that
\begin{eqnarray}
\nonumber \bar m(\bu_{2,n}')-\bar m(\bu_{1,n})&=&
 [\mathrm{E}m_n(\mZ_n,\bu_{2,n}')-
m_n(\mZ_n,\bu'_{2,n})]\\\nonumber&+&[m_n(\mZ_n,\bu'_{2,n})-f_2(\bu'_{2,n})]
\\\nonumber
   &-& [\mathrm{E}m_n(\mZ_n,\bu_{1,n})-
m_n(\mZ_n,\bu_{1,n})]\\\nonumber&-&[m_n(\mZ_n,\bu_{1,n})
-f_1(\bu_{1,n})]\\\nonumber
   &+& [f_2(\bu'_{2,n})-f_1(\bu'_{2,n})]+[f_1(\bu'_{2,n})-f_1(\bu_{1,n})]\\\nonumber
   &\geq&o_p(1,n)+o(1,\epsilon;n)O_p(1,n)= o_p(1,\epsilon,n).\nonumber
\end{eqnarray}
In other words, for any $\delta>0$,
$$\lim\limits_{n\to\infty}\P(\varliminf\limits_{\epsilon\to 0}[\bar
m(\bu_{2,n}')-\bar m(\bu_{1,n})]<-\delta)=0.$$
Similarly, $$\lim\limits_{n\to\infty}\P(\varliminf\limits_{\epsilon\to
0}[\bar m(\bu_{2,n}')-\bar
m(\bu_{2,n})]<-\delta)=0.$$ Combine these two
together, we have, for any $\delta>0$ and $r\in [0,1]$,
\begin{equation}\label{eqn:lemma2Delta}
\lim\limits_{n\to\infty}\P(\varliminf\limits_{\epsilon\to 0}\bar
m(\bu_{2,n}')-[r\bar m(\bu_{1,n})+(1-r)\bar
m(\bu_{2,n})]<-\delta)=0.
\end{equation}
Define
$$\Delta(\bu_1,\bu_2)\triangleq \max_{0\leq r\leq1} [r\bar m(\bu_1)+(1-r)\bar
m(\bu_2)-\bar m(r\bu_1+(1-r)\bu_2)],$$ and $$
 C_{\delta}=\inf\{C\ge 0|\forall\bu_1,\bu_2\in\bTheta\mbox{~and~}\|\bu_1-\bu_2\|\ge C,\Delta(\bu_1,\bu_2)>\delta\}.$$
 Since $\bTheta$ is compact, there exists $\delta_1>0,$ such that $C_{\delta}$ exists for any
 $\delta\in(0,\delta_1]$. Since $\bar m$ is strictly convex,
 $C_{\delta}\rightarrow0,$ as $\delta\rightarrow0,$ and $C_{\delta}>0.$ We thus conclude
 from (\ref{eqn:lemma2Delta}) that for any $\delta>0$
$$\lim\limits_{n\to\infty}\P(\varlimsup\limits_{\epsilon\to
0}\Delta(\bu_{1,n},\bu_{2,n})>\delta)=0.$$
By the definition of $C_{\delta}$, for $\delta'>\delta$,
$$\{\varlimsup\limits_{\epsilon\to
0}\|\bu_{1,n}-\bu_{2,n}\|>C_{\delta'}\}\subseteq\{\varlimsup\limits_{\epsilon\to
0}\Delta(\bu_{1,n},\bu_{2,n})>\delta\}.$$ Therefore, for any $\delta'>\delta>0$,
$$\lim\limits_{n\to\infty}\P(\varlimsup\limits_{\epsilon\to 0}\|\bu_{1,n}-\bu_{2,n}\|>C_{\delta'})=0.$$
Since $\delta>0$ implies $C_{\delta}>0$ and $\delta\rightarrow0$
implies $C_{\delta}\rightarrow0$, $C_\delta'$ can be arbitrarily small. Thus, for any $\eta>0$,
$$\lim\limits_{n\to\infty}\P(\varlimsup\limits_{\epsilon\to 0}\|\bu_{1,n}-\bu_{2,n}\|>\eta)=0.$$

%
%
\end{proof}

\begin{proof}[Proof of Lemma \ref{lemma:stability:strong}]
From Condition (C2) and similar to the proof of Lemma \ref{lemma::stability:weak},
\begin{eqnarray*}
  &~&\sup_{\stackrel{r\in[0,1]}{\btheta_i\in\bTheta,i=1,2}}\large\{[r
M_n(\mZ_n+\mE_n,\btheta_1)+(1-r)
M_n(\mZ_n+\mE_n,\btheta_2) \\
   &-& M_n(\mZ_n+\mE_n,r\btheta_1+(1-r)\btheta_2)]\\&-& [r
M_n(\mZ_n,\btheta_1)+(1-r)
M_n(\mZ_n,\btheta_2) \\
   &-& M_n(\mZ_n,r\btheta_1+(1-r)\btheta_2)]\large\}\\
   &=&O_p(1,n)o(1,\epsilon;n).
\end{eqnarray*}
In view of this and Condition (C5), there exists $\epsilon_0>0,$ such that if
$\|\mE_n\|\le\sqrt{n}\epsilon_0,$ then
\begin{equation}\label{eqn:lemmastrongstability:strictconvex}
\lim_{n\to\infty} P\left(M_n(\mZ_n+\mE_n,\btheta) \mbox{~is strictly
convex within~} U(\btheta_0;\delta_0)\cap\bTheta\right)=1.
\end{equation}

On the other hand, (C4) and weak asymptotic stability imply
$$\varlimsup_{\substack{n\rightarrow\infty}}\mathrm{P}(\varliminf_{\epsilon\rightarrow0}
\mathrm{diam}\bigcup_{\substack{\|\mE\|<\sqrt{n}\epsilon\\
\mZ_n+\mE_n\in \SZ}}\{\arglmin
M_n(\mZ_n+\mE_n,\btheta)\}>\delta_0/2)=0,$$
$$\varlimsup_{\substack{n\rightarrow\infty}}\mathrm{P}(\varliminf_{\epsilon\rightarrow0}
d(\btheta_0,\bigcup_{\substack{\|\mE\|<\sqrt{n}\epsilon\\
\mZ_n+\mE_n\in \SZ}}\{\arglmin
M_n(\mZ_n+\mE_n,\btheta)\})>\delta_0/2)=0.$$ Thus,
\begin{align}\label{eq::proof-lemma-2}
\varliminf_{\substack{n\rightarrow\infty}}\mathrm{P}(\varlimsup_{\epsilon\rightarrow0}
\bigcup_{\substack{\|\mE_n\|<\sqrt{n}\epsilon\\
\mZ_n+\mE_n\in \SZ}}\{\arglmin M_n(\mZ_n+\mE_n,\btheta)\}\subset
U(\btheta_0;\delta_0))=1.
\end{align}

Denote by $\tilde\btheta$ and $\btheta^*$ the minimizers of
$M_n(\mZ_n+\mE_n,\btheta)$ and $M_n(\mZ_n,\btheta)$ over $\btheta\in
U(\btheta_0;\delta_0)\cap\bTheta$, respectively. We have
\begin{eqnarray*}
  0&\geq& M_n(\mZ_n,\btheta^*)-M_n(\mZ_n,\tilde\btheta)  \\
   &=& M_n(\mZ_n+\mE_n,\btheta^*)-M_n(\mZ_n+\mE_n,\tilde\btheta)+O_p(1)o(1,\epsilon;n) \\
   &\geq& O_p(1)o(1,\epsilon;n).
\end{eqnarray*}
Thus
$$\lim_{\epsilon\rightarrow0}\|M_n(\mZ_n,\btheta^*)-M_n(\mZ_n,\tilde\btheta)\|=0,$$
which, along with strict convexity of $M_n(\mZ_n+\mE_n,\btheta)$
as in (\ref{eqn:lemmastrongstability:strictconvex}), implies
\begin{align}\label{eq::thetastartheta}
\lim_{\epsilon\rightarrow0}\|\btheta^*-\tilde\btheta\|=0.
\end{align}
Combining \eqref{eq::proof-lemma-2} and \eqref{eq::thetastartheta}, we have
%
$$\varliminf_{\substack{n\rightarrow\infty}}\mathrm{P}(\varlimsup_{\epsilon\rightarrow0}\mathrm{diam}
\bigcup_{\substack{\|\mE_n\|<\sqrt{n}\epsilon\\
\mZ_n+\mE_n\in \SZ}}\{\arglmin M_n(\mZ_n+\mE_n,\btheta)\}=0)=1.$$

\end{proof}

\begin{proof}[Proof of Lemma \ref{lemma:rate}]

First, it is obvious that $\lambda\rightarrow0$ and
$\sqrt{n}\lambda\rightarrow+\infty.$ Since $\zeta_1$ is fixed, it suffices to show
\begin{align}\label{eq::lem3-eq1}
\sqrt{n}\lambda
\exp[{(\frac{\lambda_0}{\lambda}\zeta_1-\alpha_1)^uL(\frac{\lambda_0}{\lambda}\zeta_1-\alpha_1)}]\rightarrow0.
\end{align}

From the monotonicity of $\log(\cdot)$ in $(0,\infty)$, \eqref{eq::lem3-eq1} is equivalent to
\begin{align}\label{eq::lem3-eq2}
\log(\sqrt{n}\lambda)+(\frac{\lambda_0}{\lambda}\zeta_1-\alpha_1)^uL(\frac{\lambda_0}{\lambda}\zeta_1-\alpha_1)\rightarrow-\infty.
\end{align}
Since $\log(\sqrt{n}\lambda)\to \infty$ and $L(\cdot)<0$ by assumption, we only need to show
\begin{align}\label{eq::lem3-eq3}
\log(\sqrt{n}\lambda)\ll (\frac{\lambda_0}{\lambda})^uL(\frac{\lambda_0}{\lambda}).
\end{align}

Since $L(\cdot)$ is a slowly-varing function, we have $L(\cdot)=O(1)$ or $L(\cdot)\to -\infty$. Thus, a sufficient condition will be $\log(\sqrt{n}\lambda)\ll (\frac{\lambda_0}{\lambda})^u$.
Then taking $1/u$-th power on both slides, we have  $\lambda[\log(\sqrt{n}\lambda)]^{1/u}\ll \lambda_0$, which can be guaranteed by the assumption $\lambda(\log n)^{1/u}\ll \lambda_0$.

%

\end{proof}

\begin{proof}[Proof of Lemma \ref{lemma:oracle}]
Part 1 (Parameter estimation consistency). \\
It suffices to show that for any $\epsilon>0,$ there exists
$C>0,$ such that
\begin{equation}\label{eqn:para-consistency}
\varliminf_{\substack{n\rightarrow\infty}}\P[\sup_{\|\bu\|=C}\tilde{l}(\btheta_0+
\bu(n^{-\frac{1}{2}}+p_{1,n}))-\tilde{l}(\btheta_0)<0)]>1-\epsilon.
\end{equation}
This is because, under (\ref{eqn:para-consistency}), there exists
a local maximizer
$\hat{\btheta}=(\hat{\theta}_0,\cdots,\hat{\theta}_p),$ such that
$\|\hat{\btheta}-\btheta_0\|= O_p(\alpha_n)$, where
$\alpha_n=n^{-\frac{1}{2}}+p_{1,n}$. Therefore, consistency holds
as $p_{1,n}\rightarrow 0$.

To show (\ref{eqn:para-consistency}), we apply the Taylor expansion
to get, in view of the fact that $p_{\lambda ,j}(|\beta_{j,0}
+(n^{-\frac{1}{2}}+p_{1,n})u_j|)\ge 0=p_{\lambda ,j}(|\beta_{j,0}|)$
for all $j>q$,
\begin{eqnarray}\label{eq::proof-lemma4}
   &~& \tilde{l}(\btheta_0+(n^{-\frac{1}{2}}+p_{1,n})\bu)-\tilde{l}(\btheta_0)
   \\\nonumber
   &\le&(n^{-\frac{1}{2}}+p_{1,n})\bu^T l'(\btheta_0)-\frac{1}{2}n\bu^T\mR
   \bu(n^{-\frac{1}{2}}+p_{1,n})^2(1+o_p(1))\nonumber\\
   &&-n\sum_{j=1}^q
p'_{\lambda,j}(|\beta_{j,0}|)\mbox{sgn}(\beta_{j,0})(n^{-\frac{1}{2}}+p_{1,n})u_j\nonumber\\
&&+n\sum_{j=1}^q\frac{1}{2}p''_{\lambda,j}(r_j)(n^{-\frac{1}{2}}+p_{1,n})^2u_j^2,\nonumber
\end{eqnarray}
where $|\beta_{j,0}|\leq
r_j\leq|\beta_{j,0}|+(n^{-\frac{1}{2}}+p_{1,n})C.$

We now deal with the four terms on right-hand side of
\eqref{eq::proof-lemma4}. From Condition (1.a), the first term is
of order $1+n^{\frac{1}{2}}p_{1,n}$ and the second term is
of order $(1+n^{\frac{1}{2}}p_{1,n})^2$. By the Cauchy
inequality,
$$n\sum_{j=1}^q
p'_{\lambda,j}(|\beta_{j,0}|)\mbox{sgn}(\beta_{j,0})(n^{-\frac{1}{2}}+p_{1,n})u_j
\leq n\sqrt{q}p_{1,n}(n^{-\frac{1}{2}}+p_{1,n})\|u\|.$$ Since
$q$ is fixed, it is also of the same or a smaller order compared to the
second term. As $p_{4,n}\rightarrow0$, $p''_{\lambda,j}(\cdot)$ in
the fourth term vanishes. Thus the fourth term is also controlled by the second term.
Regarding the constant involving $\bu$, the second term contains
$\|\bu\|^2$, while both the first and third terms contain $\|\bu\|$. Consequently, the
whole expression is controlled by its second term as long as we
choose a sufficiently large $C$, noting that $\mR$ is positive definite.

If Condition (1.b) holds, the fourth term is negative and of the same
sign as the second term. Therefore, our conclusion still holds.

\textmd{ Part 2 (Model Selection Consistency)}.\\
First of all, from any of Condition (2.a), (2.b), and (2.c), we have for $n$ sufficiently large, $p_{2,n}>0$.
Assume we have an $\alpha_n$-consistent local minimizer
$\hat{\btheta}_n$. If the model selection consistency does not hold, then
there exists a $j\in \{q+1,q+2,\cdots,p\}$ such that
$\hat{\beta}_j\neq0$.  This will result in contradiction if we can show
that there exist $\epsilon_n\ll\alpha_n$ and a
neighborhood of $\hat{\beta}_j$, $U(\hat{\beta}_j; \epsilon_n)$, within which  the sign
of $\frac{\partial \tilde{l}(\bbeta)}{\partial\beta_j}$ does not
change. Since $\hat{\beta_j}$ is a solution, the sign of left
derivative  should be different from the sign of
right derivative at this value. So the non-zero
$\hat{\beta_j}$ does not exist. Taking the Taylor expansion, we get
\begin{equation}\label{sp}
    \frac{\partial\tilde{l}(\btheta)}{\partial\beta_j}=\frac{\partial l(\btheta_0)}{\partial\beta_j}+
\sum_{l=1}^p\frac{\partial^2l(\bbeta_0)}{\partial\beta_l\partial\beta_j}(\beta_l-\beta_{l,0})(1+o_p(1))-
np'_{\lambda,j}(|\beta_j|)\mbox{sgn}(\beta_j),
\end{equation}
for $j=q+1,q+2,\cdots,p.$ We see that the third term on the right-hand side of \eqref{sp}
does not depend on $\bbeta_0$ and is only related to the sign
of $\beta_j$, which remains constant in $U(\hat{\beta_j};\epsilon_n)$ since $\hat\beta_j\neq 0$.
So if the first two terms are controlled by the third
one, we can derive the sparsity using the method above. The
coefficient of the sign function in the third term should be
positive to control the direction of the derivative. That is,
$\inf_{q<j\leq p}p'_{\lambda,j}(0)>0.$

Under Condition (2.a), the orders of the three terms in \eqref{sp} are
$\sqrt{n},\sqrt{n}+np_{1,n},$ $p_{2,n}(u_n),$ where $u_n$ is the
sequence given in Condition (2.a).  Condition (2.a) guarantees that
the first two terms  are controlled by the third one. Likewise,
Condition (2.c) leads to the same conclusion.

Define
\begin{align}\label{eq::lem4:h_n}
h_n(\bv)\triangleq \left\{ \frac{1}{2}\bv\mR
\bv^T+\sum_{j=1}^q\mbox{sgn}(\beta_{j,0})\frac{p'_{\lambda,j}(|\beta_{j,0}|)}{p_{1,n}}
v_j+\sum_{j=q+1}^p\frac{p'_{\lambda,j}(0)}{p_{1,n}}|v_j|\right\},
\end{align}
and
\begin{align}\label{eq::lem4:v_n}
\bv_n=\frac{\hat{\btheta}-\btheta_0}{p_{1,n}}=\arg\max_{\bv}
\tilde{l} (\btheta_0+p_{1,n}\bv).
\end{align}

 Since
$p_{3,n}(\epsilon_n^{(2)})\rightarrow0$, we have $\sup_{q<j\leq
p}|p'_{\lambda,j}(0)|$ $<\infty$ and, for any $u_n\to 0$,
\begin{equation}\label{eqn::conditionsparsity2}
p_{2,n}(|u_n|)=p_{2,n}+o(1).
\end{equation}
The above results, together with
$p_{4,n}\rightarrow0,$ imply that
\begin{eqnarray}
\nonumber
  &~&\tilde{l}(\btheta_0+p_{1,n}\bv)=\tilde{l}(\btheta_0)+p_{1,n}\bv^T\frac{\partial
l(\btheta_0)}{\partial\btheta}+\frac{1}{2}p^2_{1,n}\bv^T\frac{\partial^2
l(\btheta_0)}{\partial\btheta\partial\btheta^T}
\bv(1+o_p(1))\\\nonumber
  &-&np_{1,n}^2(1+o(1))\sum_{j=1}^p\Big[\mbox{sgn}(\beta_{j,0})v_j\frac{p'_{\lambda,j}(|\beta_{j,0}|)}{p_{1,n}}I_{\beta_{j,0}\neq0}+
  \frac{p'_{\lambda,j}(|\beta_{j,0}|)}{p_{1,n}}|v_j|
I_{\beta_{j,0}=0}\Big]\\\nonumber
&=&\tilde{l}(\btheta_0)+np_{1,n}^2\Big[-\frac{1}{2}\bv\mR
\bv^T-\sum_{j=1}^q\mbox{sgn}(\beta_{j,0})\frac{p'_{\lambda,j}(|\beta_{j,0}|)}{p_{1,n}}
v_j\\\nonumber
&&-\sum_{j=q+1}^p\frac{p'_{\lambda,j}(0)}{p_{1,n}}|v_j|\Big](1+o_p(1)).\nonumber
\end{eqnarray}

As $\bv_n$ is an maximizer given by \eqref{eq::lem4:v_n}, there exists $\alpha_n\rightarrow0,$
such that for any given $\epsilon>0,$ there exists a constant $C>0$ such that
$$\P(\sup_{\|\bv\|\le C}\tilde
l(\btheta_0+p_{1,n}(\bv_n+\alpha_n\bv))-\tilde
l(\btheta_0+p_{1,n}(\bv_n))\le 0)>1-\epsilon.$$ From
$$h_n(\bv_n+\alpha_n\bv)-h_n(v_n)=-(np_{1,n}^2)^{-1}[\tilde
l(\btheta_0+p_{1,n}(\bv_n+\alpha_n\bv))-\tilde
l(\btheta_0+p_{1,n}(\bv_n))](1+o_p(1)),$$
 we get
$$\P(\sup_{\|\bv\|\le
C}h_n(\bv_n+\alpha_n\bv)-h_n(v_n)\ge0)>1-\epsilon.$$
Thus we have
\begin{align}\label{eq::lem4:limit}
\frac{\hat{\btheta}-\btheta_0}{p_{1,n}}-\arg\min_{\bv}h_n(\bv)\xrightarrow{~~p~~}0.
\end{align}

Now, we investigate the minimizer of $h_n(\bv)$.
Let $\bv_n\triangleq(v_{0,n},\bv_{1,n}^T,\bv_{2,n}^T)^T$ where
$\bv_{1,n}^T$ is a $q\times1$ vector. We see
$\bv_{2,n}\rightarrow\bv_2=0$ in probability is a necessary
condition for sparsity. From conclusions above,
\begin{eqnarray}\label{formula:solutionv1v2}
&~&(v_{0,n},\bv_{1,n}^T,\bv_{2,n}^T)-\arg\min_{v_0,\bv_1,\bv_2}[\frac{1}{2}(v^2_0\mR_{00}+2v_0\mR_{01}\bv_1+2v_0\mR_{02}\bv_2\\\nonumber
&+&\bv_1^T\mR_{11}\bv_1+2\bv_1^T\mR_{12}\bv_2+\bv_2^T\mR_{22}\bv_2)+\bv^T_1\bLambda_1\ba+|\bv^T_2|\bLambda_2\boldsymbol1]\xrightarrow{~p~}0,
\end{eqnarray}
where
$$\bLambda_1=p_{1,n}^{-1}\mbox{diag}\{p'_{\lambda,1}(|\beta_{10}|),
p'_{\lambda,2}(|\beta_{20}|),\cdots,p'_{\lambda,q}(|\beta_{q0}|)\},
$$ $$\bLambda_2=p_{1,n}^{-1}\mbox{diag}\{p'_{\lambda,q+1}(0),
p'_{\lambda,q+2}(0),\cdots,p'_{\lambda,p}(0)\},$$
$\ba=\mathrm{sgn}(\bbeta_{1,0})$ and $|\bz|\leq1,$

KKT conditions lead to the following equations.
\begin{eqnarray}\label{eq::proof-lemma-4:KKT}
 \mR_{00}v_0+\mR_{01}\bv_1+\mR_{02}\bv_2&=&0,\\\nonumber
 \mR_{10}v_0+\mR_{11}\bv_1+\mR_{12}\bv_2+\bLambda_1\ba &=&0,\\\nonumber
 \mR_{20}v_0+\mR_{21}\bv_1+\mR_{22}\bv_2+\bLambda_2\bz &=&0,\nonumber
\end{eqnarray}
or
\begin{eqnarray}
 \nonumber
 \mR^*_{11}\bv_1+\mR^*_{12}\bv_2+\bLambda_1\ba &=&0,\\\nonumber
 \mR^*_{21}\bv_1+\mR^*_{22}\bv_2+\bLambda_2\bz &=&0,\nonumber
\end{eqnarray}
 where
 $\bz=(z_{1},\cdots,z_{p-q})^T$ and $$z_{j}\in\{z|\mbox{if
}v_{q+j}\neq0,z=\mathrm{sgn}(v_{q+j});\mbox{ else }|z|\leq1\}.$$
 Therefore
\begin{eqnarray}
\nonumber
  \bv_2=0&\Longleftrightarrow&\mR^*_{21}\mR_{11}^{*-1}\bLambda_1\ba+\bLambda_2\bz=0,|\bz|\leq\textbf{1} \\\nonumber
  &\Longleftrightarrow&
|\bLambda_2^{-1}\mR^*_{21}\mR_{11}^{*-1}\bLambda_1\ba|\leq\mbox{\textbf{1}},
\end{eqnarray}
which is a necessary condition for sparsity. Now let
$$\bLambda_2(\bv_2)\triangleq p_{1,n}^{-1}\mbox{diag} \{ p'_{\lambda,q+1}(v_{q+1}),
p'_{\lambda,q+2}(v_{q+2}),\cdots,p'_{\lambda,p}(v_{p})\}.$$ Then
\begin{eqnarray}\label{eq:lem4:2.b}
   \frac{\partial\tilde{l}(\btheta)}{\partial\bbeta_{2}}
&=&\frac{\partial\tilde{l}(\btheta_0)}{\partial\bbeta_{2}}+np_{1,n}[(\mR_{20},\mR_{21},\mR_{22})(v_0,\bv_1^T,\bv_2^T)^T
  \\ \nonumber
   &~&(1+o_p(1))-\bLambda_2(p_{1,n}\bv_2)\mbox{sgn}(\bbeta_{2})].
\end{eqnarray}
 From
(\ref{eqn::conditionsparsity2}),
$\bLambda_2(p_{1,n}\bv_2)=\bLambda_2+o(1)$. So with
$\sqrt{n}p_{1,n}\rightarrow\infty$ a sufficient condition for
sparsity is that
$(\mR_{20},\mR_{21},\mR_{22})(v_0,\bv_1^T,\bv_2^T)^T$ is controlled
by $\bLambda_2$ coordinatewisely. That is,
$$\mR_{20}v_0+\mR_{21}\bv_1+\mR_{22}\bv_2+\bLambda_2\bz=0,|\bz|<1,$$
which is similar to the last equation of
(\ref{eq::proof-lemma-4:KKT}) and equivalent to
$\mR^*_{21}\mR_{11}^{*-1}\bLambda_1\ba+\bLambda_2\bz=0,|\bz|<1$,
which is implied by $|\bLambda_2^{-1}\mR^*_{21}\mR_{11}^{*-1}\bLambda_1\ba|<\mbox{\textbf{1}}$. Furthermore, it can be guaranteed by
\begin{align}\label{eq::bound}
\|\mR_{21}^*\mR_{11}^{*-1}\|_{\infty}<(\|\bLambda_2^{-1}\|_{\infty}\|\bLambda_1\|_{\infty})^{-1}\rightarrow
C.
\end{align}
From \eqref{eq:lem4:2.b}, it follows that the first two terms are controlled by the third one.
Part 3 (Asymptotic normality and oracle property). \\
Suppose
$\hat{\btheta}=(\hat\alpha,\hat{\bbeta}_1^T,\boldsymbol0_{p-q}^T)^T$
is the local minimizer, noting that $\hat{\bbeta}_2=\boldsymbol0_{p-q}$ with
probability tending to 1. Using parameter estimation consistency and
model selection consistency property, for $j=0,1,2,\cdots,q,\ $
\begin{eqnarray}
 \nonumber
  0 &=& \frac{\partial
\tilde{l}(\hat{\btheta})}{\partial\theta_j}=\frac{\partial
l(\hat{\btheta})}{\partial\theta_j}-np'_{\lambda,j}(|\hat{\beta}_j|)\mbox{sgn}(|\hat{\beta}_j|)I_{j\neq0}
\\\nonumber
   &=& \frac{\partial
l(\btheta_0)}{\partial\theta_j}+\sum_{i=0}^q(\frac{\partial^2l(\btheta_0)}{\partial\theta_j\partial\theta_i})
(\theta_i-\theta_{i,0})(1+o_p(1))\\\nonumber
&-&I_{j\neq0}[np'_{\lambda,j}(|\beta_{j,0}|)\mbox{sgn}(\beta_{j,0})+n
p''_{\lambda,j}(|\beta_j|)(\hat{\beta}_j-\beta_{j,0})(1+o_p(1))].
 \nonumber
\end{eqnarray}
Thus
\begin{eqnarray*}
   \sqrt{n}\bSigma_2
\left[\left(
                                  \begin{array}{c}
                                    \hat\alpha-\alpha_{10} \\
                                    \hat{\bbeta}_1-\bbeta_{10} \\
                                  \end{array}
                                \right)
   +\bSigma_2^{-1}\bb\right]
   \longrightarrow N\left[0,\left(
                                                      \begin{array}{cc}
                                                        \mR_{00} & \mR_{01} \\
                                                        \mR_{10} & \mR_{11} \\
                                                      \end{array}
                                                    \right)\right].
\end{eqnarray*}
If the order of $\bb$ is controlled by $1/\sqrt{n}$, the oracle
property holds.

\end{proof}

\begin{proof}[Proof of Theorem \ref{theo:LAMPasymp}]
Using the notations of Lemma \ref{lemma:oracle} and the form of the
LAMP penalty, $p_{1,n}=\lambda
g'(\alpha_1-\frac{\lambda_0}{\lambda}\zeta_{1})/g'(\alpha_1);p_{2,n}
=\lambda;p_{2,n}(u_n) =\lambda
g'(\alpha_1-\frac{\lambda_0}{\lambda}u_n)/g'(\alpha_1);
p_{3,n}=\lambda_0g''(\alpha_1)/|g'(\alpha_1)|;$
$p_{4,n}=\lambda_0g''(\alpha_1-\frac{\lambda_0}{\lambda}\zeta_1)/|g'(\alpha_1)|.$

From (C6) and the Cauchy inequality,
$$\mathrm{E}|Y\bX^{T}\btheta-g(\bX^{T}\btheta)|<\infty, \forall
\btheta\in \bTheta.$$ Thus, by law of large numbers,
$m_n(\btheta)=-l_n(\btheta)\rightarrow \bar m(\btheta) $ in
probability. Recall
$$\lambda_{\min}(\E \bX g''(\bX^{T}\btheta)\bX^{T})>0, \forall
\btheta\in \bTheta,$$
from (C6),
$-\mathrm{E}l(\btheta)$ is a strictly convex
function of $\btheta.$ (C1) holds.

 From smoothness of the function $g$ and
compactness of $\bTheta$, (C2) holds.

 From (C7), (C8),
$\lambda_0/\lambda\rightarrow \infty,$ together with
$\lim_{\xi\rightarrow\infty}g'(\xi)=0,\alpha_1$ is a constant and
$\lambda\rightarrow 0,$ we have
$$\sup_{x\ge0}p'_{\lambda}(x)=\sup_{x\ge0}\lambda\frac{g'(\alpha_1-\lambda_0/\lambda x)}{g'(\alpha_1)}\rightarrow0.$$
(C3) holds, and $p_{1,n}\rightarrow0,p_{4,n}\rightarrow0.$

 From (C6)
\begin{align}
\nonumber E\left[\|\bX\|_2^2g''(\bX^T\btheta_0)+\|
\bX\|_{1}^3\sup_{\|\btheta-\btheta_0\|\leq \delta} g'''(
\bX^T\btheta)\right]<\infty,
\end{align}
which leads to (C10)-(C12) in Lemma \ref{lemma:oracle}. From $p_{1,n}\rightarrow0, p_{4,n}\rightarrow0,$ Condition (1.a) in
Lemma \ref{lemma:oracle} holds. (C4) holds.

Given $u_n=O(p_{1,n}+1/\sqrt{n})$, from
$\sqrt{n}p_{1,n}\rightarrow0,$ we get $u_n=O(1/\sqrt{n}).$ From
$\sqrt{n}\lambda\rightarrow\infty,$ we get
$u_n/\lambda\rightarrow0$. Together with the smoothness of the
function $g$, $p_{2,n}(u_n)/p_{2,n}\rightarrow1.$ From Condition
(C8), it is obvious that
$$p_{2,n}>0,\frac{p_{2,n}}{p_{1,n}}\rightarrow\infty,\sqrt{n}p_{2,n}\rightarrow\infty.$$
Condition (2.a) and conditions for asymptotic normality in Lemma
\ref{lemma:oracle} hold. Lastly, (C5) is implied by (C9).

In conclusion, with (C6)$\sim$(C9), the penalized
maximum likelihood estimator based on the LAMP family is consistent
and asymptotically normal, and achieves model selection consistency
and strong asymptotic stability.
\end{proof}
\begin{proof}[Proof of Theorem \ref{theo:algorithmconverge}]
The idea of the proof is adapted from \cite{IRLS06} and \cite{onestep}. For
convenience, we define the following notations.
\begin{eqnarray}
 \nonumber
 M_n(\btheta)&\triangleq &M_{(1)}(\btheta)+M_{(2)}(\btheta),\\\nonumber
M_{(1)}(\btheta)&\triangleq &-\frac{1}{n}l(\btheta),
M_{(2)}(\btheta)\triangleq
\sum_{j=1}^pp_{\lambda}(|\theta_j|),\\\nonumber
  \btheta^{(\bm)}&\triangleq &(\theta_0^{(m_0)},\theta_1^{(m_1)},\cdots,\theta_p^{(m_p)})^T,\\\nonumber
  \btheta^{(\bm)}_{(-j)}(x)&\triangleq &(\theta_0^{(m_0)},\theta_1^{(m_1)},\cdots,\theta_{j-1}^{(m_{j-1})},x,
\bTheta_{j+1}^{(m_{j+1})},\theta_p^{(m_p)})^T,\\\nonumber Q(x, j,
\btheta^{(\bm)}) &\triangleq&
a(\btheta^{(\bm)}_{(-j)}(x),\btheta^{(\bm)}),\\\nonumber \phi(x, j,
\btheta^{(\bm)})&\triangleq&\sum_{1\leq k\leq p}
p_{\lambda}(|\theta_k^{(m_k)}|)+p'_{\lambda}(|\theta_j^{(m_j)}|)
(|x|-|\theta_j^{(m_j)}|),\\\nonumber R(x, j,
\btheta^{(\bm)})&\triangleq&Q(x, j, \btheta^{(\bm)}) +\phi(x, j,
\btheta^{(\bm)}),
\end{eqnarray}
 where $m_0,m_1,\cdots,m_p$ are the
iteration times of $\theta_0,\theta_1,\cdots,\theta_p$ respectively
and $\bm=(m_0,m_1,\cdots,m_p)$. From the concavity of the penalty on
the positive part, we have
$$M_{(2)}(\btheta^{(\bm+\be_{j})})\leq\phi(\theta_j^{(m_j+1)},j, \btheta^{(\bm)}),$$
where $\be_j$ is a $(p+1)$-vector with the $j$-th element 1
and all the others 0; from conditions of the theorem, we see
$$M_{(1)}(\btheta^{(\bm+\be_{j})})\leq
Q(\theta_j^{(m_j)},j, \btheta^{(\bm)}),$$ and
$$M(\btheta^{(\bm+\be_{j})})\leq
R(\theta_j^{(m_j+1)},j, \btheta^{(\bm)}).$$
In the algorithm,
$$\theta_j^{(m_j+1)}=\mathrm{arg}\min_{\theta}R(\theta,j, \btheta^{(\bm)}).$$
Then
\begin{eqnarray}
 \nonumber
  M(\btheta^{(\bm)})&=&R(\theta_j^{(m_j)},j, \btheta^{(\bm)})\\\nonumber
   &>& R(\theta_j^{(m_j+1)},j, \btheta^{(\bm)})\\\nonumber &\geq& M(\btheta^{(\bm+\be_{j})}).
\end{eqnarray}
Since $M$ decreases as iteration continues and has an
lower bound, it converges.

By monotonicity of $M(\btheta^{(t)})$ ($t$ represents the number of
iterations), all points $\btheta^{(\bm)}$ are in a compact set
$$\{\btheta\in \mathcal{C}|M(\btheta)\leq M(\btheta^{(\boldsymbol
0_{p+1})})\}.$$ It is compact because $M(\btheta)$
is continuous and coercive. Then there exists a convergent
subsequence $\btheta^{(t_l)}$, $\btheta^*\in\mathcal{C},$
$j_0\in\{0,1,\cdots,p\}$ such that
$$\lim_{l\rightarrow\infty}\btheta^{(t_l)}=\btheta^*;j^{(t_l+1)}\equiv
j_0.$$ Next let $m^*_l\triangleq m_{j_0}^{(t_{l}+1)}.$
 For any $v\in \mathcal{C}$, we have:
$$M(\btheta^{(t_{l+1})})\leq M(\btheta^{(t_l+1)})\leq R(\theta_{j_0}^{(m^*_l)},j_0|\btheta^{(t_l)})
\leq R(v_{j_0},j_0|\btheta^{(t_l)}).$$ Assume
$M(\btheta^{t_{l+1}})\rightarrow
M(\btheta^{*})=R(\btheta^{*},j|\btheta^{*}),
j=0,\cdots,p$. Taking limit $l\rightarrow\infty$ on both
sides of the equation above, we have
$$M(\btheta^*)\leq
\lim_{l\rightarrow\infty}R(v_{j_0},j_0|\btheta^{(t_l)})=R(v_{j_0},j_0|\btheta^*).$$
Thus the subgradient of $R(\cdot,j_0|\btheta^{*})$  at
$\btheta^{*}$ is 0, which is exactly the derivative of
$M(\btheta)$ with respect to $\bTheta_{j_0}$ at $\btheta^{*},$ because it can
be easily verified that the smooth approximation keeps the first-order
derivative the same.

From the algorithm and the definition of the ``viol'' function,
$$\forall j\in\{0,\cdots,p\},|\dot R(\cdot,j|\btheta^{(t_l)})|\leq |\dot R(\cdot,j_0|\btheta^{(t_l)})|.$$
Taking the derivative of both sides of the equation above, the
subgradient of $R(\cdot,j|\btheta^{*})$ at $\btheta^{*}$ is
0, which is exactly the same as the partial derivative of $M(\btheta)$ with respect to
$\bTheta_{j}$ at $\btheta^{*},\forall j\in\{0,\cdots,p\}.$ From
strict convexity, $\btheta^*$ is the unique local minimum of
$M(\btheta)$.

Method 1 for logistic regression uses Lemma 1 in \cite{IRLS06}. That
is,
$$\log(1+e^z)\leq\frac{z}{2}+\log(e^{-z_0/2}+e^{z_0/2})+\frac{1}{2}\frac{\tanh(.5z_0)}{z_0}(z^2-z_0^2).$$

\end{proof}

\end{document}